\documentclass[aps,prl,twocolumn,superscriptaddress,groupedaddress,floatfix]{revtex4}
\usepackage[utf8]{inputenc}
\usepackage[T1]{fontenc}
\usepackage{graphicx}
\usepackage{booktabs}
\usepackage{tabularx}
\usepackage{nicefrac}
\usepackage{algpseudocode}
\usepackage[dvipsnames,svgnames]{xcolor}
\usepackage{caption}

\usepackage{amsmath, amssymb, bm}
\usepackage{placeins}
\usepackage{amsthm} 

\newtheorem{lemma}{Lemma}
\newtheorem{theorem}{Theorem}

\usepackage[colorlinks=true,allcolors=blue]{hyperref}

\begin{document}

\title{Can homophily explain public underestimation of climate policy support?}


\author{Ekaterina Landgren}
\affiliation{Department of Environmental Social Sciences, Stanford Doerr School of Sustainability, Stanford University, Stanford, CA 94305, USA}
\affiliation{Cooperative Institute for Research in Environmental Sciences, University of Colorado Boulder, Boulder, CO 80309, USA}
\email{kath.landgren@stanford.edu}

\author{Shriya Nagpal}
\affiliation{Department of Mathematics, Pitzer College, Claremont, CA 91711, USA}

\author{Joshua Garland}
\affiliation{Global Security Initiative Center on Narrative, Disinformation and Strategic Influence, Arizona State University, Tempe, AZ 85287, USA}

\author{Yaw Acquah}
\affiliation{Department of Mathematics, Pitzer College, Claremont, CA 91711, USA}

\author{Matthew G. Burgess}
\affiliation{Department of Economics, College of Business, University of Wyoming, Laramie, WY 82071, USA}

\begin{abstract}
Many climate policies enjoy majority U.S. public support, yet both Republicans and Democrats underestimate that support. Homophily is one possible explanation for this pattern: selective exposure to like-minded views could lead opponents to underestimate support more than supporters. We use a stochastic block model and preferential attachment model to test how homophily and social network structure shape misperception of policy support. Homophily alone explains opponents underestimating support more than supporters, but supporters only underestimate when their own homophily is low enough that they disproportionately associate with opponents. When Bayesian rescaling and inaccurate priors are added to the models, homophily must still be highly asymmetric to explain observed misperception patterns. Adding media bias to the model produces realistic misperception even with symmetric homophily. However, empirical evidence on media bias in coverage of climate change policy is mixed. Our results offer theoretical grounding for understanding misperception of public opinion broadly.
\end{abstract}

\maketitle

\section{Introduction}

Climate change is one of the most polarizing issues in the United States (U.S.), but it is less polarizing than the public believes. Many climate policies enjoy support from two-thirds or more of the U.S. public \cite{marlon2022change,burgess2024climate}. A recent study, however, \cite{sparkman2022americans} found that 80\textendash90\% of Americans underestimate this climate policy support, often by 20 percentage points or more. Similar misperception patterns have also been found in other countries \cite{andre2024globally}. Public overestimation of political polarization is widespread across many issues and measures (e.g., refs. \cite{levendusky2016mis,druckman2022mis}). 

Public underestimation of climate policy support is an example of pluralistic ignorance, whereby people who hold the majority opinion (climate policy supporters, in this case) underestimate their opinion's popularity \cite{geiger2016}. Pluralistic ignorance about climate change could be consequential. For example, it could make individuals discuss climate change or policies addressing it less often with their peers \cite{geiger2016,mildenberger2019}. It could also cause politicians and their staff to underestimate public support for climate change policies, and as a result advance such policies less or oppose them more \cite{hertel2019legislative}.

Why Americans are so misinformed on their peers' views of climate change is a major puzzle in climate politics. In their study documenting this misperception, Sparkman et al. \cite{sparkman2022americans} find it highest among Republicans (who oppose climate policy at higher rates \cite{egan2024us,burgess2024climate}) and those with high exposure to right-leaning news media, on average. These observations suggest two hypotheses explaining misperception: homophily among policy opponents and media bias. 

Homophily describes the tendency of like-minded people to preferentially associate with each other \cite{mcpherson2001birds}. If people opposed to climate change policies preferentially associate with each other, views within their social networks would skew against climate policy (compared to the general public), causing them to underestimate public support for climate policy \cite{dixon2024complexity}. This is known as false consensus \cite{ross1977false}. On the other hand, homophily could also cause people who support climate policy to \textit{overestimate} public support for the same reason. The conditions under which homophily among policy opponents could outweigh homophily among supporters to produce aggregate misperception are poorly understood.

There are at least two ways media bias might cause people to underestimate public support for climate policy. First, media could overrepresent negative views of climate policy, compared to public opinion. This could cause misperception via availability bias \cite{tversky1973availability} or other mechanisms by which elite communications and cues influence public opinion and perception \cite{gerbner1986living,sherman2024connections,samuelson1988status}. 

However, real-world evidence for such media bias is mixed. Boykoff and Boykoff~\cite{boykoff2004balance} famously found a similar phenomenon for climate science in the 1990s and early 2000s, which they called `balance as bias'. Media were overrepresenting views doubting the physical facts regarding human contributions to climate change---compared to the level of consensus among scientists---in an attempt to carry out the journalistic norm of ``balanced'' reporting. A 2021 follow-up study~\cite{mcallister2021balance} found that false balance about climate science was no longer widespread, outside of a few right-leaning media outlets. More recently, Landgren et al.~\cite{landgren2026mediabalance} assessed balance of television coverage of climate policy, compared to public opinion. They found individual outlets provided either overwhelmingly positive coverage (e.g., CNN) or overwhelmingly critical coverage (Fox News). However, they did not find clear evidence of aggregate bias towards critical coverage, compared to public opinion, across all outlets nor in individuals' self-reported media diets. 

Second, media could exaggerate or emphasize polarization, either by emphasizing polarization explicitly or by foregrounding extreme and charged minority views on both sides~\cite{chinn2020politicization,garrett2019partisan}. If this biased people towards perceiving evenly split (50:50) public opinion, it could cause them to underestimate the popularity of popular policies or positions (such as climate policies with two-thirds support), and overestimate the popularity of unpopular policies or positions. However, media bias exaggerating polarization is difficult to measure.

Cognitive processing of uncertainty could also cause people to bias perceptions of policy support towards more even splits (50:50), because uncertainty causes people to overestimate the sizes of small proportions and underestimate the sizes of large ones via Bayesian reasoning \cite{guay2025quirks}. Regardless of the cause, there are documented instances of perception biases away from extreme values. For example, support for political violence tends to be overestimated \cite{mernyk2022correcting} and support for popular forms of gun control tends to be underestimated \cite{dixon2020public}.

Here, we use analytical and simulation analyses of two agent-based social-network models to assess the potential for homophily to cause empirically realistic patterns of public underestimation of climate policy support, both alone and in combination with social network structure, Bayesian reasoning under uncertainty (via Bayesian rescaling and inaccurate priors), and media bias. We compare our models' predictions to empirical results of Sparkman et al. \cite{sparkman2022americans}, who measured misperception at the individual level in a large survey ($N = 6{,}119$) that included several relevant covariates such as media diet and opinions of climate policy.

Our models suggest that homophily alone cannot explain observed misperception patterns, unless the homophily is extremely (and unrealistically) asymmetric. Opponents of climate policy would have to be highly insular, and supporters would need such low homophily that they actually associated disproportionately with opponents.
The empirical evidence on possible asymmetric homophily is mixed. There is some evidence that more conservative social media users exhibited higher homophily~\cite{colleoni2014echo, boutyline2017social}. On the other hand, a recent study suggests that Democrats are more isolated than Republicans in both offline and online settings~\cite{brown2025relationship}.
Adding Bayesian reasoning and inaccurate priors lessens---but does not eliminate---the asymmetric homophily required to explain observed misperception patterns.  Combining realistic homophily with media bias towards overrepresenting negative opinions of climate policy could explain observed misperception patterns. However, there is not clear empirical evidence that such bias exists \cite{landgren2026mediabalance}. Our results thus do not resolve the puzzle of climate policy misperception, but they shed light on conditions under which homophily can and cannot produce pluralistic ignorance and false consensus---insights which apply beyond this particular setting. We discuss how future theoretical and empirical research could help to further resolve the puzzle of climate policy misperception.

\section{Results}

\subsection{Empirical patterns of misperception}

\begin{figure*}[t]
\centering
\includegraphics[width=0.7\linewidth]{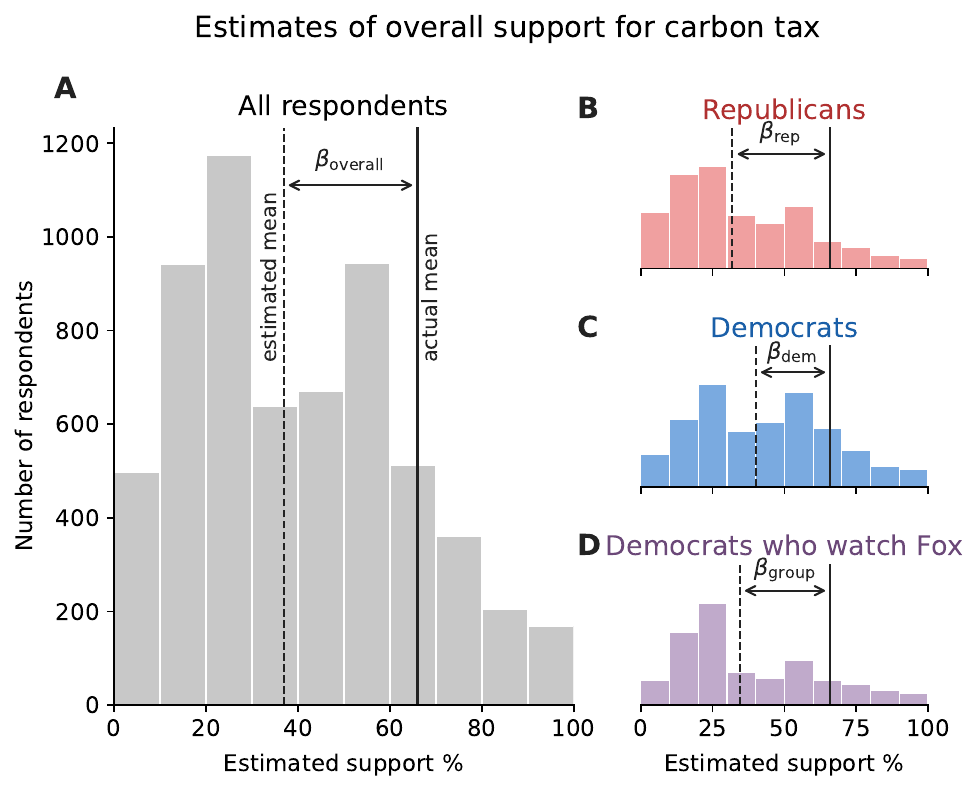}
\caption{Distribution of estimates of public support for a carbon tax from Sparkman et al.'s survey \cite{sparkman2022americans}.
    Respondents were asked to estimate the percentage of Americans who support a carbon tax for fossil fuel companies.
    \textbf{(A)} Distribution of estimates across all survey respondents (10-percentage-point bins).
    The solid vertical line marks actual public support (67\%); the dashed line marks the mean estimated support (37\%).
    The gap between actual support and mean estimate is misperception $\beta_{\mathrm{overall}}$.
    (\textbf{B}--\textbf{D}) Distributions broken down by partisan group and media diet: \textbf{(B)} Republicans, \textbf{(C)} Democrats, and \textbf{(D)} Democrats who report watching Fox News at least once per week.
    Across all groups, respondents substantially underestimated public support for a carbon tax, with Republicans showing the largest misperception and Democrats who watch Fox News showing greater misperception than Democrats overall.}
\label{fig:histograms}
\end{figure*}

To establish an empirical baseline for comparison with our models, we show key patterns of observed misperception \cite{sparkman2022americans} among all Americans, partisan subgroups, and frequent viewers (at least once per week) of Fox News, a right-leaning media outlet whose viewership is less accepting of and concerned about climate change than the general public, on average \cite{feldman2012climate,sparkman2022americans}. Fig.~\ref{fig:histograms} shows misperception of support for a carbon tax.  Misperception levels for four climate policies surveyed by Sparkman et al. \cite{sparkman2022americans} (carbon tax for fossil fuel companies, 100\% renewable energy by 2035, siting RE on public lands, and the Green New Deal), and for whether or not one is `worried' about climate change, show qualitatively similar patterns (See Supporting Information (SI) Tab.~\ref{table:all-demographics}). 

Misperception of support is highest among Republicans and regular Fox News viewers (25\textendash40\% across policies). Conservative news media consumption is positively correlated with average misperception (Fig.~\ref{fig:correlation}). However, average misperception exceeds 20\% within all demographic subgroups Sparkman et al. \cite{sparkman2022americans} surveyed (e.g., by age, race, gender, and education level, in addition to party). This suggests that explanations for misperception of climate policy support are incomplete if they only apply to conservatives, Fox News viewers, or policy opponents.

\subsection{Model}
\subsubsection{Measuring misperception}

We begin by defining the opinion state $s(i) \in \{0,1\}$ for each node $i$, where $s(i)=1$ indicates support for the policy and $s(i)=0$ indicates opposition. The true fraction of supporters in the population is thus
\[
f_s=\frac{1}{N}\sum_{i=1}^{N}s(i),
\]
with $f_o = 1-f_s$ representing the fraction of opponents.

Consider a simple, undirected, unweighted network $G=(V,E)$ with $|V|=N$ nodes and adjacency matrix $A \in \mathbb{R}^{N \times N}$. Network structure matters because individuals do not observe the full population; instead, they form estimates based on their local neighborhood. Specifically, we assume individual $i$ estimates the level of support by averaging over their own opinion and those of their immediate neighbors:
\[
\hat{f}_{s,i}
=
\frac{
s(i)
+
\sum_{j=1}^{N}A_{ij}s(j)
}
{1+\sum_{j=1}^{N}A_{ij}}.
\]

We define individual $i$'s misperception as
\begin{equation}
\beta_{i}=f_s - \hat{f}_{s,i},
\label{eq:1}
\end{equation}
which is positive when support is underestimated and negative when overestimated.

For any group of size $N_{\text{group}}$, the average misperception is
\begin{equation}
\beta_{\text{group}}=\frac{1}{N_{\text{group}}}\sum_{i \in \text{group}} {\beta}_i.
\label{eq:beta_group}
\end{equation}
The quantities of primary interest are $\beta_s$ (misperception among supporters) and $\beta_o$ (misperception among opponents). Substituting the expression for $\hat{f}_{s,i}$ into Eq.~\eqref{eq:beta_group} yields the network-dependent formula
\[
\beta_{\mathrm{group}}(A)
=
f_s
-
\frac{1}{N_{\text{group}}}
\sum_{i \in \text{group}}
\frac{
s(i)
+
\sum_{j=1}^{N}A_{ij}s(j)
}
{1+\sum_{j=1}^{N}A_{ij}}.
\]
Figure~\ref{fig:cartoon} (right) illustrates this mechanism on a path graph with $N=3$.

\subsubsection{Agent-based social-network models}
Our analysis uses two agent-based network models: a preferential attachment model and a stochastic block model.
Both models are social network models; how networks are constructed is the main difference between them. In both models, the agents' opinions influence tie formation. The agents in both models can be thought of as either representing individuals or collective entities such as news outlets or their accounts on a social media site. We are deliberately ambiguous in defining an agent, since the modern information environment blurs the line between an individual with a large platform and a media outlet. An edge between two agents represents that the agents know each other's climate policy views and opinions. We assume that agents use this information to estimate the prevalence of climate policy support in the general public. Thus, if the agents connected to an agent of interest have a different average level of support for the policy, then the agent of interest will misperceive public support. For simplicity, we limit our analysis to undirected edges.

Fig.~\ref{fig:cartoon}A shows conceptual representations of the two models considered in this study. Both models incorporate homophily, but differ in the network formation mechanism, and consequently network structure. In the stochastic block model, homophily parameters directly define the probability of edge formation between agents, producing modular communities. In contrast, in the preferential attachment model with homophily~\cite{karimi2018homophily}, homophily parameters modulate attachment to high-degree agents, producing high degree heterogeneity among the agents.

\begin{figure*}[t]
\centering
\includegraphics[width=0.5\linewidth]{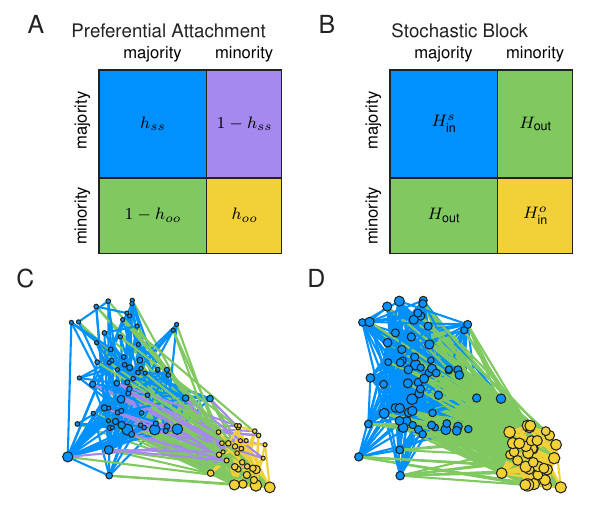}
\raisebox{0.5cm}{%
    \includegraphics[width=0.4\linewidth]{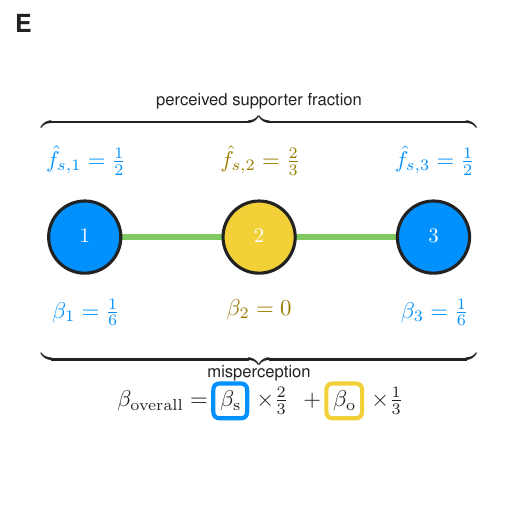}
   }

\caption{\textbf{Differences in homophily implementation between stochastic block model and preferential attachment model.} In the stochastic block model (left), majority\textendash majority edges are governed by the parameter $H_\text{in}^s$, minority\textendash minority edges are governed by the parameter $H_\text{in}^o$, and majority\textendash minority edges are governed by the parameter $H_\text{out}$. These block-level probabilities directly determine the expected density of ties in each quadrant of the mixing matrix. In contrast, in the preferential attachment model (middle), homophily operates through group-specific attachment biases rather than fixed block probabilities. The minority homophily parameter $h_{oo}$ governs the fraction of the ties allocated to minority \textendash minority ties in the configuration model compared to all minority ties, and likewise for  the majority homophily parameter $h_{ss}$ and majority\textendash majority edges. The resulting networks exhibit more pronounced degree heterogeneity under preferential attachment than under the stochastic block model, as shown in the bottom panels. On the right, we show a misperception example for a three-node line graph, following Eq.~\eqref{eq:beta_group}, where individual estimates $\hat{f}_{s,i}$ combine to produce the group-level misperception $\beta_{\text{group}}$.
}
\label{fig:cartoon}
\end{figure*}

\medskip
\noindent
\textbf{Preferential attachment model}

\noindent
The first model we consider uses the well-known principle of preferential attachment to construct the social network. Agents are more likely to connect to other agents that already have many connections \cite{barabasi1999emergence}. This creates a heavy-tailed degree distribution among agents.

In order to capture a range of plausible connection patterns between supporters and opponents, we add two tunable homophily parameters, following Karimi et al. \cite{karimi2018homophily}. Adding homophily makes the probability of connection between agents depend on whether they hold the majority opinion (e.g., supporting climate policy) or the minority opinion. We consider the potentially asymmetric case where a supporter homophily parameter $h_{ss}$ regulates the probability of connection between supporters, and opponent homophily $h_{oo}$ regulates the probability of connection between opponents. Between-group connections are regulated by complementary probabilities: from supporters to opponents $h_{so}=1-h_{ss}$ and from opponents to supporters $h_{os}=1-h_{oo}$. The direction of edge formation only matters during the generation of the network from a configuration model, and the generating process produces an undirected network.

We can express the average perceived fraction of supporters, by opponents and supporters respectively, as follows. Here $p_{oo}$, $p_{ss}$, $p_{os}$, and $p_{so}$ denote the probabilities of opponent--opponent, supporter--supporter, opponent--supporter, and supporter--opponent connections respectively, derived from $h_{ss}$, $h_{oo}$, and the minority growth factor $C$ (see \nameref{appendix:BA}):
\begin{align}
\hat{f}_s^{o} &= 
\frac{\tfrac{f_s}{f_o}\,p_{so} + p_{os}}{2p_{oo} + \tfrac{f_s}{f_o}\,p_{so} + p_{os}}
\label{eq:BA-average-perceived-o} \\[1ex]
\hat{f}_s^{s} &= 
\frac{2\,p_{ss}}{2p_{ss} + \tfrac{f_o}{f_s}\,p_{os} + p_{so}}
\label{eq:BA-average-perceived-s}
\end{align}

We then weight supporter misperception $\beta_s = f_s - \hat{f}_s^{s}$ and opponent misperception $\beta_o = f_s - \hat{f}_s^{o}$ by their respective fractions of the population, $f_s$ and $f_o = 1 - f_s$, to obtain overall misperception:
\begin{equation}
\beta_{\text{overall}} = f_s \,\beta_s + (1 - f_s)\,\beta_o.
\end{equation}

\medskip
\noindent \textbf{Stochastic block model}
\noindent
Our second model is a stochastic block model, which assigns edges between agents with probabilities based on which groups (`blocks') agents belong to. In our case, the blocks are supporters and opponents of a (climate) policy, and the probability of two agents being connected---given which group each agent is in---is governed by parameters related to homophily.  

To formalize this, suppose there are $N$ agents in a network, and define
\[
S := \{1, \dots, \lfloor f_s N \rfloor\}
\]
to be the set of agents supporting the (climate) policy, and 
\[
O := \{\lfloor {f_s}N\rfloor + 1, \dots, N\}
\]
to be the set of agents opposing the policy. To define their connections, consider the function $W:[0,1]^2 \rightarrow[0,1]$ defined by:

\[
W(x, y)=
\begin{cases}
H_{\text{in}}^{s} & \text { if }x, y \in[0,f_s]\\
H_{\text{in}}^{o} & \text { if }x, y \in(f_s,1]\\
H_{\text{out}} & \text { if }x \in[0,f_s], y \in(f_s,1] \text { or }x \in(f_s,1], y \in[0,f_s]
\end{cases}
\]
where $0< H_{\text{in}}^{s},H_{\text{in}}^{o},H_{\text{out}}\leq1$. Here, $H_{\text{in}}^{s}$ and $H_{\text{in}}^{o}$ denote the connection probabilities, respectively, for supporters and opponents with agents from their own group; and $H_{\text{out}}$ denotes the connection probability for supporters and opponents with each other. 

Let $A^N$ denote the adjacency matrix of the network. The network is then generated by the following process:
\[
P\!\left(A_{ij}^N = 1\right) = W_{ij}^N := W\!\left(\frac{i}{N}, \frac{j}{N}\right),
\]
for all \( i \neq j \in \{1, \ldots, N\} \) and $A_{ij}^N = 0$ if $A_{ij}^N \neq 1$. Further, we set $A_{ji}^N =A_{ij}^N$ for $i\neq{j}$. Thus, \( A^{N} \) is a realization of a two-block stochastic block model with within-block connection probability \( H_{\text{in}}^{s},H_{\text{in}}^{o} \) and cross-block connection probability \( H_{\text{out}} \). 

In this model, we consider three network-based measures of misperception. Recall that for any group of individuals $U \subseteq V$, the group misperception is defined as

\[
\beta_U(A^N)
=
\frac{1}{|U|}
\sum_{i\in U}
\beta_i(A^N).
\]

The \textit{overall misperception} is obtained by taking \(U=V\), the \textit{supporters' misperception} by taking \(U=S\), and the \textit{opponents' misperception} by taking \(U=O\). Here, the true level of support in the population is assumed to be $f_s$. Finally, we consider the weighted (expectation) network \(\bar{W}^N\), where \(\bar{W}^{N}_{ij}=W^{N}_{ij}\) for \(i\neq j\) and $0$ otherwise. We define the corresponding misperception measures \(\beta_U(\bar{W}^N)\) analogously to \(\beta_U(A^N)\).

We prove the following theorem, which shows that the misperception measures on the random network \(A^N\) converge almost surely to their deterministic counterparts on the weighted network \(\bar W^N\).

\begin{theorem}
With probability 1,
\begin{itemize}
    \item[(a)] $\lim_{N\rightarrow{\infty}} \left|(\beta_V({A}^{N})-G\right|=0$, where
    \begin{align*}
    G &= f_s-\Bigg(\left(f_s\right) \frac{f_s H_{\text {in }}^s}{f_s H_{\text {in }}^s+\left(1-f_s\right) H_{\text {out }}}\\
    &\quad+\left(1-f_s\right) \frac{f_s H_{\text {out }}}{f_s H_{\text {out }}+\left(1-f_s\right) H_{\text {in }}^o}\Bigg).
    \end{align*}
    \item[(b)] $\lim_{N\rightarrow{\infty}} \left|\beta_S({A}^{N})-\hat{G}\right|=0$, where
    \begin{align*}
    \hat{G}&=  f_s -  \frac{f_sH_{\text{in}}^{s}}{f_sH_{\text{in}}^{s}+(1-f_s)H_{\text{out}}}, \text{and}
    \end{align*}
    \item[(c)]  $\lim_{N\rightarrow{\infty}} \left|\beta_O({A}^{N})-\bar{G}\right|=0$, where
    \begin{align*}
    \bar{G}&=  f_s -  \frac{f_sH_{\text{out}}}{f_sH_{\text{out}} + (1-f_s)H_{\text{in}}^{o}}.
    \end{align*}
\end{itemize}
\end{theorem}

\begin{proof}
    See \nameref{sec:Extended-Proofs-for-SBM} in Supplementary Information and Extended Data.
\end{proof}

\subsection{Misperception patterns in the model}

\begin{figure*}[t]
\centering
\includegraphics[width=0.67\linewidth]{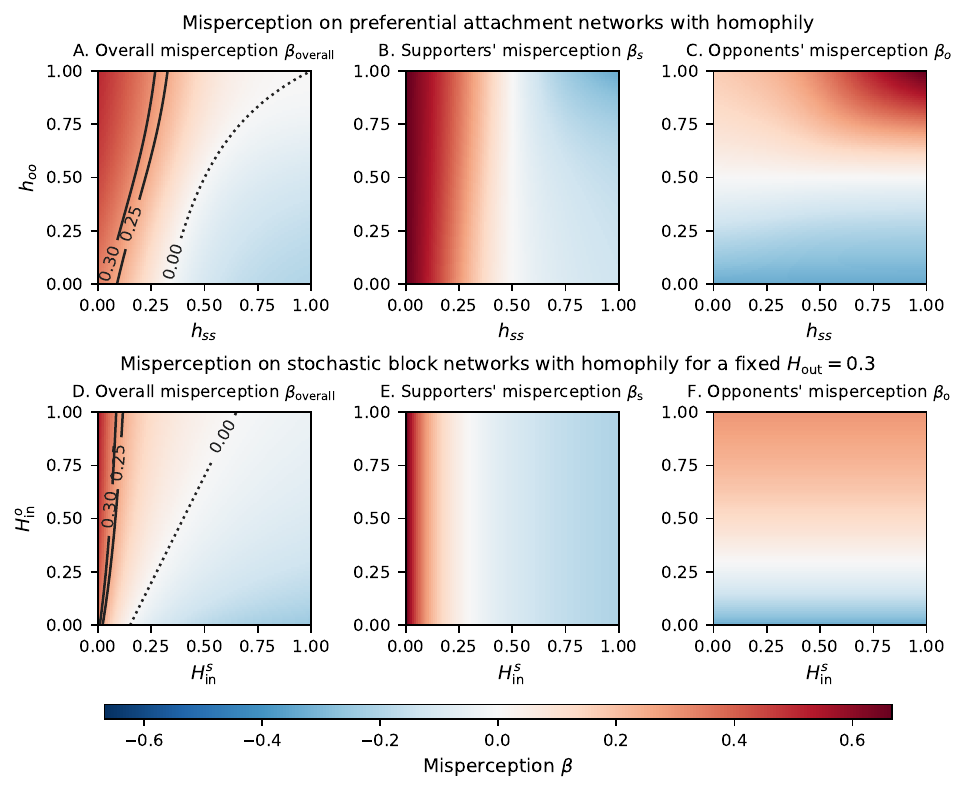}

\caption{
\textbf{Misperception as a function of network homophily in Preferential Attachment and Stochastic Block networks.}
Each row shows overall misperception ($\beta_{\text{overall}}$, the weighted average of group-level misperception), supporters' misperception ($\beta_{s}$), and opponents' misperception ($\beta_{o}$), respectively.
Panels~A--C show the Preferential Attachment network case, where within-group attachment probabilities for the minority ($h_{oo}$) and majority ($h_{ss}$) groups vary across the grid.
Panels~D--F show the Stochastic Block network case, where within-group homophily for supporters ($H_{\text{in}}^{s}$) and opponents ($H_{\text{in}}^{o}$) varies across the grid; the out-group edge probability $H_{\text{out}} = 0.3$ is held fixed.
Contour lines in Panels~A and~D indicate misperception levels consistent with empirical misperception of public support for climate policy~\cite{sparkman2022americans}. Dashed line indicates the boundary between positive and negative misperception.
Across both models, misperception is maximized when homophily among opponents exceeds that among supporters.}

\label{fig:no_transform}
\end{figure*}

We assess which homophily parameter regions give rise to patterns of misperception most consistent with empirical data. Using the analytical results derived in the previous sections, we show misperception in the preferential attachment model (Fig.~\ref{fig:no_transform} top row) and the stochastic block model (Fig.~\ref{fig:no_transform} bottom) for a fixed supporter fraction $f_s=\frac{2}{3}$.

In the preferential attachment model, we independently vary supporter homophily $h_{ss}$ and opponent homophily $h_{oo}$. Supporter misperception $\beta_{s}$ increases monotonically as supporter homophily $h_{ss}$ decreases (Fig.~\ref{fig:no_transform}B). The regime where $\beta_{s}>0$, which corresponds to the underestimation of the fraction of supporters, occurs when supporter homophily $h_{ss}<0.5$, indicating weaker ties between supporters than between supporters and opponents. In contrast, opponent misperception  $\beta_{o}$ increases monotonically as opponent homophily $h_{oo}$ increases (Fig.~\ref{fig:no_transform}C). The regime where $\beta_{o}>0$ occurs when opponent homophily $h_{oo}>0.5$, indicating that stronger ties between opponents than between opponents and supporters are necessary to create misperception. Since there are more supporters than opponents in our model, overall misperception $\beta_{\text{overall}}$ tends to be positive when supporter homophily is low. We observe that the band corresponding to the empirically observed range ($\beta_{\text{overall}} \in [0.25, 0.30]$) covers approximately 6\% of the full $(h_{ss}, h_{oo})$ parameter space, and overall misperception is higher above the diagonal, where $h_{oo}>h_{ss}$, indicating stronger connections between opponents than between supporters (Fig.~\ref{fig:no_transform}A).

In other words, for misperception to occur among both supporters and opponents, opponents must exhibit moderate-to-strong homophily and supporters must exhibit homophily so weak that they actually prefer to associate with opponents. In the more realistic case where supporters prefer to associate among themselves ($h_{ss}>0.5$), aggregate misperception requires opponents to have strong enough homophily and underestimation of support to outweigh supporters' weaker overestimation of support and greater numbers.

In the stochastic block model, three parameters are needed to specify the regime. For ease of comparison between the two models, we hold $H_{\text{out}}$, the probability of connection between supporters and opponents, fixed, and vary supporter homophily $H_{\text{in}}^s$ and opponent homophily $H_{\text{in}}^o$.

Misperception patterns are similar to the preferential attachment case, with supporter misperception $\beta_{s}$ increasing monotonically as $H_{\text{in}}^s$ decreases, and opponent misperception $\beta_{o}$ increasing monotonically as $H_{\text{in}}^o$ increases. Overall misperception is again highest when supporter homophily is low and opponent homophily is high. Supporters must again preferentially associate with opponents ($H_{\text{in}}^s< H_{\text{out}}$) to underestimate support ($\beta_{s}>0$) (Fig.~\ref{fig:no_transform}E), and opponents must exhibit homophily ($\beta_{o}>0$ when $H_{\text{in}}^o>H_{\text{out}}$) (Fig~\ref{fig:no_transform}F). Compared to the preferential attachment model, asymmetric homophily among opponents and/or preferential outgroup association among supporters would need to be more extreme to produce observed misperception patterns in the stochastic block model: the empirical band covers only approximately 3\% of the $(H_{\text{in}}^s, H_{\text{in}}^o)$ parameter space (with $H_{\text{out}}=0.3$ fixed), roughly half the fraction of the preferential attachment model.

\subsection{Model extensions and mechanism exploration}

The results above place strong constraints on purely homophily-based explanations for the observed misperception patterns.
To explain observed patterns of misperception of climate policy support with homophily alone, the above results suggest that we need large asymmetries in homophily, with opponents clustering tightly while supporters actively seek out opponents, exhibiting outgroup preference. Given that this scenario is implausible, homophily, while likely present, would need to act in combination with additional mechanisms in order to explain the observed misperception patterns.

To explore this possibility, we consider two model extensions. First, we implement Bayesian rescaling: a cognitive mechanism that adjusts the individuals' perceptions of the popular support according to plausible psychological processes that lead to overestimation of small proportions and underestimation of large proportions under uncertainty~\cite{guay2025quirks}. We also allow for inaccurate prior beliefs. Second, we simulate a biased information environment (e.g., in media). We simulate this in our model via artificial promotion of minority opinion (opposition) nodes into central positions (i.e., more connected and influential) in the network. 

\subsubsection{Bayesian rescaling}
\label{sec:bayesian-rescaling}
Recent literature shows that the cognitive processes governing proportion estimation under uncertainty might explain widespread misperception of group sizes, with individuals systematically overestimating small proportions and underestimating large ones~\cite{guay2025quirks}. We incorporate this additional cognitive complexity into our models as post-hoc rescaling of the information the agents receive from the network, which depends on the level of uncertainty and on agents' preconceptions about the proportions. Within a Bayesian setting, we explore the effects on misperception that arise from interaction between homophily, prior expectations, and the level of uncertainty.

Suppose $\delta \in(0, \infty)$, $\gamma \in(0,1]$. We consider the transformation 
$$\hat{f}_{s,i} = \Psi(f_{s,i})=\frac{\delta^{1-\gamma} f_{s,i}^\gamma}{\delta^{1-\gamma} f_{s,i}^\gamma+(1-f_{s,i})^\gamma},$$

\noindent where $f_{s,i}=\frac{s(i)+\sum_{j=1}^{N}A_{ij}s(j)}{1+\sum_{j=1}^{N}A_{ij}}$, i.e., the proportion of node $i$’s neighbors (including itself) that are supporters \cite{guay2025quirks}. The parameter $\gamma$ represents uncertainty, where $\gamma=1$ represents complete trust in current observation, and values of $\gamma$ close to 0 indicate more emphasis on the prior information rather than current observation. The parameter $\delta$ represents the prior in odds, $\delta = \dfrac{f_s}{f_o}$. The rescaling function $\Psi(f_{s,i})$ is shown in Fig.~\ref{fig:bayesian_rescaling}A. Note that when the observation is equal to the prior, there is no rescaling. Otherwise, proportions below the prior are overestimated and proportions above the prior are underestimated.

\begin{figure*}[t]
\centering
\begin{center}
\includegraphics[width=0.67\linewidth]{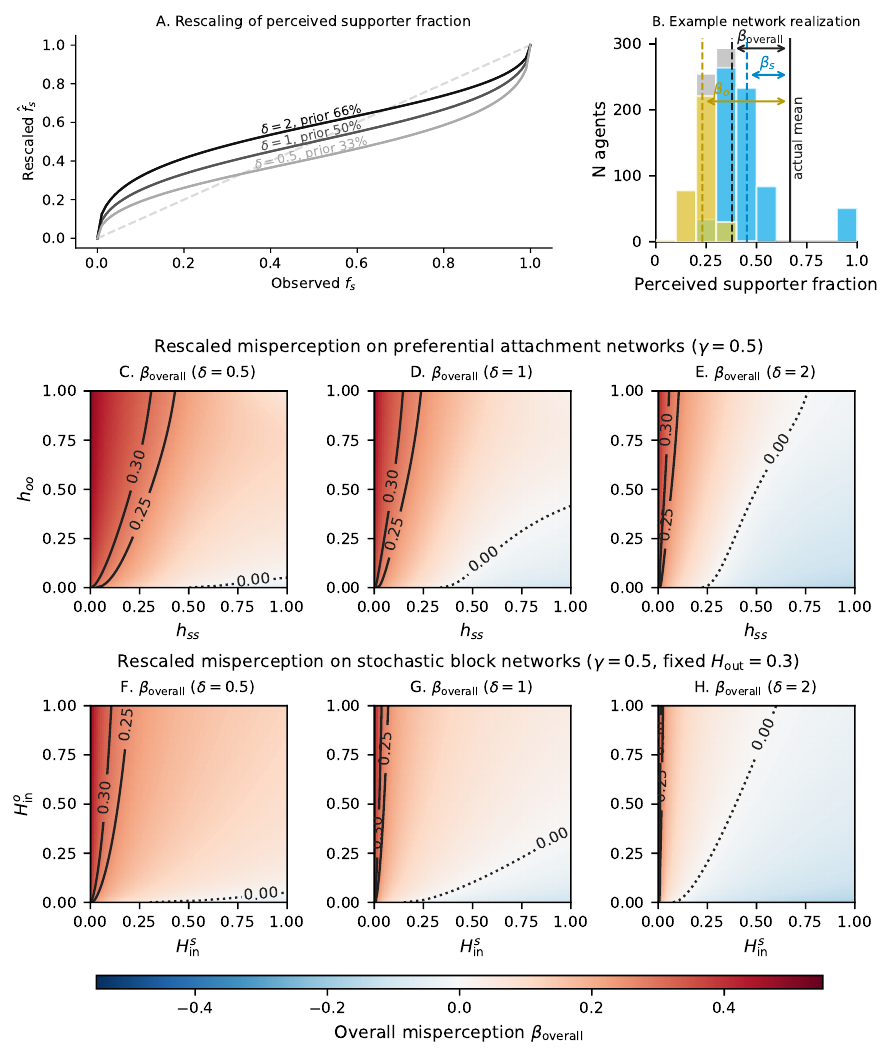}
\end{center}

\caption{\textbf{Effect of Bayesian rescaling on overall misperception.}
  (A) Bayesian rescaling of the perceived supporter fraction $\hat{f}_s$ as a
  function of the observed fraction $f_s$, for three values of the prior parameter
  $\delta$ at fixed $\gamma = 0.5$.
  $\delta =0.5$ shifts perceptions toward underestimation;
  $\delta = 1$  corresponds to a neutral prior;
  $\delta = 2$ corresponds to an accurate prior ($f_s=2/3$).
  The dashed diagonal shows the identity.
  (B) Distribution of Bayesian-rescaled perceived supporter fractions across all agents
  in an example preferential attachment network ($N=1,000$, $f_s = 2/3$, $h_{ss} = h_{oo} = 0.7$, $\delta = 0.1$, $\gamma = 0.5$).
  Grey bars: full population; blue bars: supporters ($f_s = 2/3$);
  yellow bars: opponents ($f_o = 1/3$).
  The solid vertical line marks true supporter prevalence; dashed lines mark each
  group's mean perceived fraction.
  Arrows indicate group-level misperception $\beta_s$, $\beta_o$, and overall
  misperception $\beta_\text{overall}$.
  (C--E) Overall misperception $\beta_\text{overall}$ on preferential attachment
  networks as a function of supporter in-group homophily $h_{ss}$ and opponent
  in-group homophily $h_{oo}$, for $\delta = 0.5$, $1$, and $2$ ($\gamma = 0.5$).
  Black contour lines mark $\beta_\text{overall} = 0.25$ and $0.30$, corresponding
  to the approximate values of empirically observed misperception of climate policy
  support~\cite{sparkman2022americans}; the dotted contour marks
  $\beta_\text{overall} = 0$.
  (F--H) Overall misperception on stochastic block networks ($\gamma = 0.5$,
  $H_\text{out} = 0.3$ fixed) as a function of supporter in-group homophily
  $H^{s}_\text{in}$ and opponent in-group homophily $H^{o}_\text{in}$,
  for $\delta = 0.5$, $1$, and $2$.
  Across both network types and all $\delta$ values, misperception is highest when
  opponent homophily exceeds supporter homophily.}
\label{fig:bayesian_rescaling}
\end{figure*}

\medskip

\noindent 

\medskip

\noindent \textbf{Bayesian rescaling for the preferential attachment model.} We can compute the rescaling for average supporters and average opponents by simply adjusting the average perceptions in Eq.~\eqref{eq:BA-average-perceived-o} and Eq.~\eqref{eq:BA-average-perceived-s}:
\begin{align}
\beta_{\text{overall}} ={}& f_o \left( f_s -
\Psi \left( \frac{\tfrac{f_s}{f_o}\,p_{so} + p_{os}}{2p_{oo} + \tfrac{f_s}{f_o}\,p_{so} + p_{os}}\right) \right) \nonumber\\
&{}+ f_s  \left(f_s - \Psi \left(
\frac{2\,p_{ss}}{2p_{ss} + \tfrac{f_o}{f_s}\,p_{os} + p_{so}}\right) \right).
\end{align}

Here, the rescaling function $\Psi$ captures systematic distortions in proportion estimation, rescaling the average composition of the nodes' nearest neighbors according to the prior estimates and uncertainty.

Similarly to the non-rescaled case, in order to account for the effects of degree heterogeneity within the network, we conduct numerical simulations presented in (SI Figs.~\ref{fig:si-network-size}--\ref{fig:si-density}).

\medskip
\noindent

\textbf{Bayesian rescaling for the stochastic block model.} Here we compute the misperception for agents on stochastic block networks. The overall misperception, supporters' misperception, and opponents' misperception, under this uncertainty based rescaling, are given by:

\[
\beta_{V}^\Psi(A^{N}) = f_s - \frac{1}{N} \sum_{i=1}^{n}\hat{f}_{s,i},
\]
\[
\beta_S^{\Psi}(A^{N}) = f_s - \frac{1}{|S|} \sum_{i=1}^{\lfloor {f_s}N\rfloor} \hat{f}_{s,i}, 
\]
\[
\beta_{O}^{\Psi}(A^{N}) = f_s - \frac{1}{|O|} \sum_{i=\lfloor {f_s}N\rfloor+1}^{n}\hat{f}_{s,i}.
\]

Next, let:

\begin{align*}
q_S&:=\frac{f_s H_{\text{in}}^{s}}{f_s H_{\text{in}}^{s}+(1-f_s) H_{\text{out}}}, \\
q_O&:=\frac{f_s H_{\text{out}}}{f_s H_{\text{out}}+(1-f_s) H_{\text{in}}^{o}}
\end{align*}

and consider the following transformations:

\[G^{\Psi}:=f_s-\left(f_s \Psi\left(q_S\right)+(1-f_s) \Psi\left(q_O\right)\right),\]
\[\widehat{G}^{\Psi}:=f_s-\Psi\left(q_S\right),\] 
\[\bar{G}^{\Psi}:=f_s-\Psi\left(q_O\right).\]

\begin{theorem} With probability 1

\begin{itemize}
    \item[(a)] $\lim _{N \rightarrow \infty}\left|\beta_V^{\Psi}\left(A^{N}\right)-G^{\Psi}\right|=0$,
    \item[(b)] $\lim _{N \rightarrow \infty}\left|\beta_{S}^{\Psi}\left(A^{N}\right)-\widehat{G}^{\Psi}\right|=0$, \textbf{and}
    \item[(c)] $\lim _{N \rightarrow \infty}\left|\beta_{O}^{\Psi}\left(A^{N}\right)-\bar{G}^{\Psi}\right|=0$.
\end{itemize}
\end{theorem}

\begin{proof}
    The proof is analogous to that of Theorem 1, with the local proportions replaced by their transformed values under $\Psi$.
\end{proof}

\noindent Analogously to Theorem 1, Theorem 2 defines the limiting behavior of misperception as network size grows, with the added Bayesian rescaling assumption.

\medskip
\noindent
\textbf{Misperception patterns with Bayesian rescaling.} In the preferential attachment model, inaccurate priors and Bayesian rescaling allow for realistic misperception patterns under less-extreme parameter values (Fig. \ref{fig:bayesian_rescaling}C\textendash E). For example, if people have an inaccurate prior belief of 50:50-split support (Fig. \ref{fig:bayesian_rescaling}D), it is possible for both homophilic supporters and opponents to underestimate support under a wide range of parameter values. Uncertainty and Bayesian rescaling can even allow for homophilic supporters to underestimate support under realistic priors ($\delta = 2$). However, aggregate underestimation of support by more than 20\% (as observed by Sparkman et al. \cite{sparkman2022americans}) still requires outgroup preference among supporters (Fig. \ref{fig:bayesian_rescaling}), which is unrealistic. While the magnitude of misperception differs from the mean-field approximation due to degree heterogeneity, the qualitative trends are preserved across network sizes and densities (SI Figs.~\ref{fig:si-network-size}--\ref{fig:si-density}).

The stochastic block model exhibits a qualitatively similar pattern: introducing Bayesian rescaling makes the assumptions required for realistic misperception patterns less extreme (Fig.~\ref{fig:bayesian_rescaling}F) than in the base case (Fig.~\ref{fig:no_transform}D).

\subsubsection{Media bias simulated via artificial promotion of minority nodes}
To simulate effects of media bias, we introduce a mechanism in the preferential attachment
model that promotes opposing (minority) views in the network by swapping the positions of
the highest-degree majority (supporting) node and the lowest-degree minority (opposing) node.
With a small number of swaps ($<10\%$) in a network where both opponents and supporters share
high homophily ($h_{ss} = h_{oo} = 0.75$), misperception resembles the real-world pattern
(Fig.~\ref{fig:histograms}): supporters and opponents both underestimate climate policy
support, and opponents underestimate more (Fig.~\ref{fig:swaps_effect}).
There is also a bimodal distribution in misperception (Fig.~\ref{fig:swaps_effect}), as
observed in reality~\cite{sparkman2022americans} (Fig.~\ref{fig:histograms}).
This result holds across the full range of within-group homophily parameters and when
the homophily of the supporter and opponent groups is allowed to differ
(SI Figs.~\ref{fig:asymmetric_swaps}--\ref{fig:asymmetric_diagram}).
Note that this swapping mechanism relies on degree variation among nodes and is therefore
specific to preferential attachment networks; it is not directly applicable to stochastic
block models, which do not feature such variation.

\begin{figure*}[t]
\centering

\includegraphics[width=0.9\linewidth]{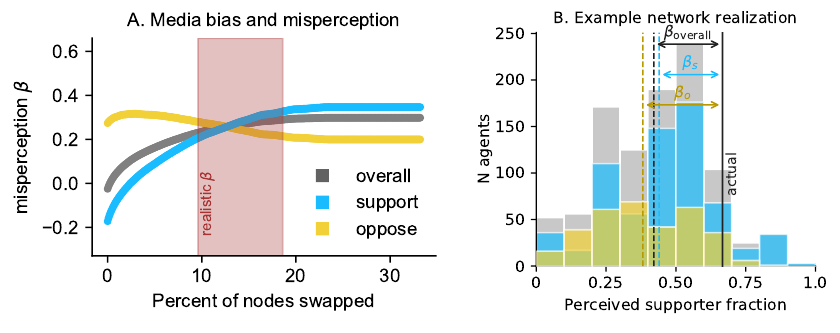}
\caption{Panel A shows the relationship between the fraction of nodes swapped and misperception $\beta$, for a network of size $10^3$, averaged over $10^3$ simulations. The starting homophily is $h=0.75$, corresponding to a network with realistically plausible high homophily. The fraction of nodes swapped is a proxy for the over-representation of minority opinion among best-connected nodes, e.g. through media bias. Panel B shows the distribution for the estimated opinion for a network with 10\% of nodes swapped for a single representative simulation of a network of $10^3$ agents.  Note the network exhibits misperception among both supporters and opponents.}
\label{fig:swaps_effect}
\end{figure*}

\section{Discussion}

Public underestimation of climate policy support is a major puzzle in climate politics \cite{sparkman2022americans,andre2024globally}, with parallels to other issues \cite{dixon2020public,mernyk2022correcting}. U.S. adults across the political spectrum underestimate public support for climate policies on the whole, and Republicans and Fox News viewers underestimate support more than other Americans \cite{sparkman2022americans}. 

Our results suggest that homophily cannot explain these misperception patterns on its own. Homophily could explain why policy opponents underestimate support, but it cannot explain why policy supporters underestimate support. Homophily could explain aggregate levels of misperception, but only if the homophily is assumed to be highly asymmetric, with policy opponents being more homophilic than supporters. Adding Bayesian rescaling to the model can explain supporter and overall misperception patterns observed, but again only with highly asymmetric homophily and/or inaccurate prior beliefs. We find that symmetric homophily can explain the observed misperception patterns when combined with media bias---simulated in our study by manipulating the network centrality of views opposing the policy.

The empirical evidence on asymmetric homophily is mixed. A key research limitation here is that homophily is most easily measured among politically active individuals online, who are a highly unrepresentative sample of the general population \cite{robertson2024inside}. Some studies have found politically conservative social media users and bloggers to more selectively associate with each other \cite{boutyline2017social,adamic2005political}, and liberal users to be exposed to a broader range of online and cultural content from across the political spectrum \cite{rogers2022liberals,chang2025liberals}. On the other hand, recent studies have also found liberals to more likely than conservatives to avoid friendships and romantic partnerships with, or online-block, people that they disagreed with politically \cite{martel2024blocking,norman2025can,sleiman2025sleeping}.

Empirical evidence on media bias is also mixed. Media outlets have historically been somewhat biased towards balanced coverage of the physical science of climate change, in contrast to broad scientific consensus \cite{boykoff2004balance}. This bias has attenuated over the past two decades, but persists in some right-leaning outlets \cite{mcallister2021balance}. However, a recent study of television news coverage of climate policy \cite{landgren2026mediabalance} did not find opposing views to be aggregately overrepresented, compared to public opinion. Instead, the authors found evidence for polarized coverage---with overwhelmingly positive coverage of climate policies by left-leaning outlets, and overwhelmingly negative coverage by right-leaning outlets. These findings were limited to television, and there could be media biases in other media. However, the patterns this study found are not consistent with the aggregate media bias necessary in our model to explain the observed misperception patterns.

In the context of our model, it is difficult to assess empirical evidence for inaccurate priors without introducing a potential tautology. Misperception itself could indicate inaccurate priors (e.g., towards more evenly divided opinions than there actually are). However, considering observed misperception as evidence of inaccurate priors, and then using the inaccurate priors to explain observed misperception in our model, is tautological. Nonetheless, Americans' documented tendency to overestimate polarization \cite{chinn2020politicization,garrett2019partisan} could cause them to have priors that generally underestimate support for popular policies, including climate policies \cite{landgren2026mediabalance}. 

Future research should explore other possible explanations of public underestimation of climate policy support. Such explanations might include, for example, self-silencing by policy supporters \cite{dixon2024complexity}, and the low salience of climate change compared to some other issues \cite{burgess2024climate}. Another proposed explanation is that opponents of climate policy have higher rates of ego-projection, or egocentric bias, than supporters, based on the partisan patterns with regard to false consensus and false uniqueness~\cite{zimmaro2025meta}. It is, of course, also possible that the levels of misperception observed by Sparkman et al. \cite{sparkman2022americans} are unrepresentative. Replication studies could help to resolve what misperception patterns exist and how they have or have not changed over time.

We focus our discussion on the issue of climate change, but our analysis could shed light on any other issue where there is public misperception of opinion. For example, Americans tend to underestimate support for gun control \cite{dixon2020public} and overestimate support for political violence \cite{mernyk2022correcting}. How homophily---asymmetric or not---could interact with media bias and inaccurate priors to produce misperception patterns is no less relevant to understanding these issues than climate policies. 

Our analysis has several limitations worth noting. To isolate the role of homophily, our models assume that individuals estimate the prevalence of support by sampling the opinions of their immediate neighbors, although alternative frameworks consider more complex forms of social sampling~\cite{dalege2025networks, zimmaro2025meta}. In both the stochastic block and preferential attachment models, individuals hold binary opinions (support or oppose) and infer public opinion solely from local observations. We also assume that homophily is fixed. In practice, homophily and perceptions likely co-evolve, as individuals may form or dissolve social ties in response to their perceptions of public opinion.

Additionally, our models assume that individuals infer public opinion from the local average of observed opinions, whereas real-world perceptions are likely shaped by a broader set of informational sources. In the Bayesian rescaling analysis, we further simplify by assuming that agents share common priors and uncertainty levels. Finally, we model social ties as undirected, treating perceived opinions between connected agents symmetrically. This assumption may be less appropriate for settings involving asymmetric information sources, such as media systems, where influence predominantly flows in one direction.

Our comparisons of model predictions to observed misperception patterns from surveys also come with caveats regarding survey research. Survey assessments of policy support can be sensitive to framing and context. For example, framing climate policies in terms of job creation and socioeconomic impact increases bipartisan support \cite{diamond2022whose}. Partisan tribalism can make people reject climate policies they would otherwise support if they believed the opposing party was proposing the policies \cite{van2018psychological}. Survey respondents may also not report their views/opinions faithfully due to factors such as social desirability bias---whereby respondents bias their responses to appear in a more socially favorable light \cite{vesely2020social}. We also note that our quantitative benchmarks derive from a single empirical study~\cite{sparkman2022americans}; if this study overestimates the typical magnitude of misperception due to sampling, timing, or methodological factors, then our results may overpredict the strength of mechanisms required to generate observed patterns.

Understanding why members of the public tend to underestimate support for climate change policies is an important puzzle in climate politics, with parallels to many other issues. Underestimation of support could affect policy uptake, and perhaps even support itself \cite{mastroianni2022widespread}. Our results suggest that homophily could play a role in promoting observed misperception, but only in combination with media bias, inaccurate priors, or some other mechanism with a similar effect biasing all views, rather than only the views of policy opponents. The effects of homophily on opinion formation and perception, in combination with other mechanisms, deserve further study.

\section{Methods}
\label{sec:mat-methods}
\subsection{Simulation details}
Simulations of networks with a majority-minority split were constructed using the algorithm described in ref. \cite{karimi2018homophily}. The algorithm for the simulation of node swaps is described in SI Methods.

\subsection{Data Availability} For part of this work, we used previously published, nationally representative panel survey data in the U.S. about perceptions of public support for climate change policies \cite{sparkman2022americans}.
Code is available on GitHub.\footnote{\url{https://github.com/kathlandgren/homophily-misperception}}

\section*{Acknowledgements}
We thank the members of the Burgess lab for comments. E.L. was supported by the Cooperative Institute for Research in Environmental Sciences Visiting Fellows Program, funded by NOAA Cooperative Agreement NA22OAR4320151, as well as by the Stanford Doerr School of Sustainability Dean's Sustainability Leaders Fellowship.

\section*{Author contributions statement} E.L., S.N., J.G., and M.G.B. designed research. E.L., S.N., and Y.A. performed research. E.L., S.N., J.G., and M.G.B. analyzed data. E.L., S.N., and M.G.B. wrote the paper. All co-authors reviewed the data and manuscript drafts.

\section*{Additional information}

The authors declare no competing interests.

\bibliography{main}

\begin{thebibliography}{45}
\expandafter\ifx\csname natexlab\endcsname\relax\def\natexlab#1{#1}\fi
\expandafter\ifx\csname bibnamefont\endcsname\relax
  \def\bibnamefont#1{#1}\fi
\expandafter\ifx\csname bibfnamefont\endcsname\relax
  \def\bibfnamefont#1{#1}\fi
\expandafter\ifx\csname citenamefont\endcsname\relax
  \def\citenamefont#1{#1}\fi
\expandafter\ifx\csname url\endcsname\relax
  \def\url#1{\texttt{#1}}\fi
\expandafter\ifx\csname urlprefix\endcsname\relax\def\urlprefix{URL }\fi
\providecommand{\bibinfo}[2]{#2}
\providecommand{\eprint}[2][]{\url{#2}}

\bibitem[{\citenamefont{Marlon et~al.}(2022)\citenamefont{Marlon, Wang,
  Bergquist, Howe, Leiserowitz, Maibach, Mildenberger, and
  Rosenthal}}]{marlon2022change}
\bibinfo{author}{\bibfnamefont{J.~R.} \bibnamefont{Marlon}},
  \bibinfo{author}{\bibfnamefont{X.}~\bibnamefont{Wang}},
  \bibinfo{author}{\bibfnamefont{P.}~\bibnamefont{Bergquist}},
  \bibinfo{author}{\bibfnamefont{P.~D.} \bibnamefont{Howe}},
  \bibinfo{author}{\bibfnamefont{A.}~\bibnamefont{Leiserowitz}},
  \bibinfo{author}{\bibfnamefont{E.}~\bibnamefont{Maibach}},
  \bibinfo{author}{\bibfnamefont{M.}~\bibnamefont{Mildenberger}},
  \bibnamefont{and}
  \bibinfo{author}{\bibfnamefont{S.}~\bibnamefont{Rosenthal}},
  \bibinfo{journal}{Environ. Res. Lett.} \textbf{\bibinfo{volume}{17}},
  \bibinfo{pages}{124046} (\bibinfo{year}{2022}).

\bibitem[{\citenamefont{Burgess et~al.}(2024)\citenamefont{Burgess, Suarez,
  Dancer, Watkins, and Marshall}}]{burgess2024climate}
\bibinfo{author}{\bibfnamefont{M.~G.} \bibnamefont{Burgess}},
  \bibinfo{author}{\bibfnamefont{C.}~\bibnamefont{Suarez}},
  \bibinfo{author}{\bibfnamefont{A.}~\bibnamefont{Dancer}},
  \bibinfo{author}{\bibfnamefont{L.~J.} \bibnamefont{Watkins}},
  \bibnamefont{and} \bibinfo{author}{\bibfnamefont{R.~E.}
  \bibnamefont{Marshall}}, \bibinfo{journal}{A Center for Social and
  Environmental Futures Report}  (\bibinfo{year}{2024}).

\bibitem[{\citenamefont{Sparkman et~al.}(2022)\citenamefont{Sparkman, Geiger,
  and Weber}}]{sparkman2022americans}
\bibinfo{author}{\bibfnamefont{G.}~\bibnamefont{Sparkman}},
  \bibinfo{author}{\bibfnamefont{N.}~\bibnamefont{Geiger}}, \bibnamefont{and}
  \bibinfo{author}{\bibfnamefont{E.~U.} \bibnamefont{Weber}},
  \bibinfo{journal}{Nature communications} \textbf{\bibinfo{volume}{13}},
  \bibinfo{pages}{4779} (\bibinfo{year}{2022}).

\bibitem[{\citenamefont{Andre et~al.}(2024)\citenamefont{Andre, Boneva, Chopra,
  and Falk}}]{andre2024globally}
\bibinfo{author}{\bibfnamefont{P.}~\bibnamefont{Andre}},
  \bibinfo{author}{\bibfnamefont{T.}~\bibnamefont{Boneva}},
  \bibinfo{author}{\bibfnamefont{F.}~\bibnamefont{Chopra}}, \bibnamefont{and}
  \bibinfo{author}{\bibfnamefont{A.}~\bibnamefont{Falk}},
  \bibinfo{journal}{Nat. Clim. Change} pp. \bibinfo{pages}{1--7}
  (\bibinfo{year}{2024}).

\bibitem[{\citenamefont{Levendusky and Malhotra}(2016)}]{levendusky2016mis}
\bibinfo{author}{\bibfnamefont{M.~S.} \bibnamefont{Levendusky}}
  \bibnamefont{and} \bibinfo{author}{\bibfnamefont{N.}~\bibnamefont{Malhotra}},
  \bibinfo{journal}{Public Opinion Quarterly} \textbf{\bibinfo{volume}{80}},
  \bibinfo{pages}{378} (\bibinfo{year}{2016}).

\bibitem[{\citenamefont{Druckman et~al.}(2022)\citenamefont{Druckman, Klar,
  Krupnikov, Levendusky, and Ryan}}]{druckman2022mis}
\bibinfo{author}{\bibfnamefont{J.~N.} \bibnamefont{Druckman}},
  \bibinfo{author}{\bibfnamefont{S.}~\bibnamefont{Klar}},
  \bibinfo{author}{\bibfnamefont{Y.}~\bibnamefont{Krupnikov}},
  \bibinfo{author}{\bibfnamefont{M.}~\bibnamefont{Levendusky}},
  \bibnamefont{and} \bibinfo{author}{\bibfnamefont{J.~B.} \bibnamefont{Ryan}},
  \bibinfo{journal}{The Journal of Politics} \textbf{\bibinfo{volume}{84}},
  \bibinfo{pages}{1106} (\bibinfo{year}{2022}).

\bibitem[{\citenamefont{Geiger and Swim}(2016)}]{geiger2016}
\bibinfo{author}{\bibfnamefont{N.}~\bibnamefont{Geiger}} \bibnamefont{and}
  \bibinfo{author}{\bibfnamefont{J.~K.} \bibnamefont{Swim}},
  \bibinfo{journal}{Journal of Environmental Psychology}
  \textbf{\bibinfo{volume}{47}}, \bibinfo{pages}{79} (\bibinfo{year}{2016}),
  ISSN \bibinfo{issn}{0272-4944},
  \urlprefix\url{https://www.sciencedirect.com/science/article/pii/S027249441630038X}.

\bibitem[{\citenamefont{Mildenberger and Tingley}(2019)}]{mildenberger2019}
\bibinfo{author}{\bibfnamefont{M.}~\bibnamefont{Mildenberger}}
  \bibnamefont{and} \bibinfo{author}{\bibfnamefont{D.}~\bibnamefont{Tingley}},
  \bibinfo{journal}{British Journal of Political Science}
  \textbf{\bibinfo{volume}{49}}, \bibinfo{pages}{1279–1307}
  (\bibinfo{year}{2019}).

\bibitem[{\citenamefont{Hertel-Fernandez
  et~al.}(2019)\citenamefont{Hertel-Fernandez, Mildenberger, and
  Stokes}}]{hertel2019legislative}
\bibinfo{author}{\bibfnamefont{A.}~\bibnamefont{Hertel-Fernandez}},
  \bibinfo{author}{\bibfnamefont{M.}~\bibnamefont{Mildenberger}},
  \bibnamefont{and} \bibinfo{author}{\bibfnamefont{L.~C.}
  \bibnamefont{Stokes}}, \bibinfo{journal}{American Political Science Review}
  \textbf{\bibinfo{volume}{113}}, \bibinfo{pages}{1} (\bibinfo{year}{2019}).

\bibitem[{\citenamefont{Egan and Mullin}(2024)}]{egan2024us}
\bibinfo{author}{\bibfnamefont{P.~J.} \bibnamefont{Egan}} \bibnamefont{and}
  \bibinfo{author}{\bibfnamefont{M.}~\bibnamefont{Mullin}},
  \bibinfo{journal}{PS: Political Science \& Politics}
  \textbf{\bibinfo{volume}{57}}, \bibinfo{pages}{30} (\bibinfo{year}{2024}).

\bibitem[{\citenamefont{McPherson et~al.}(2001)\citenamefont{McPherson,
  Smith-Lovin, and Cook}}]{mcpherson2001birds}
\bibinfo{author}{\bibfnamefont{M.}~\bibnamefont{McPherson}},
  \bibinfo{author}{\bibfnamefont{L.}~\bibnamefont{Smith-Lovin}},
  \bibnamefont{and} \bibinfo{author}{\bibfnamefont{J.~M.} \bibnamefont{Cook}},
  \bibinfo{journal}{Annual review of sociology} \textbf{\bibinfo{volume}{27}},
  \bibinfo{pages}{415} (\bibinfo{year}{2001}).

\bibitem[{\citenamefont{Dixon et~al.}(2024)\citenamefont{Dixon, Clarke,
  Jacquet, Evensen, and Hart}}]{dixon2024complexity}
\bibinfo{author}{\bibfnamefont{G.}~\bibnamefont{Dixon}},
  \bibinfo{author}{\bibfnamefont{C.}~\bibnamefont{Clarke}},
  \bibinfo{author}{\bibfnamefont{J.}~\bibnamefont{Jacquet}},
  \bibinfo{author}{\bibfnamefont{D.~T.} \bibnamefont{Evensen}},
  \bibnamefont{and} \bibinfo{author}{\bibfnamefont{P.~S.} \bibnamefont{Hart}},
  \bibinfo{journal}{Commun. Earth Environ.} \textbf{\bibinfo{volume}{5}},
  \bibinfo{pages}{76} (\bibinfo{year}{2024}).

\bibitem[{\citenamefont{Ross et~al.}(1977)\citenamefont{Ross, Greene, and
  House}}]{ross1977false}
\bibinfo{author}{\bibfnamefont{L.}~\bibnamefont{Ross}},
  \bibinfo{author}{\bibfnamefont{D.}~\bibnamefont{Greene}}, \bibnamefont{and}
  \bibinfo{author}{\bibfnamefont{P.}~\bibnamefont{House}}, \bibinfo{journal}{J.
  of Exp. Soc. Psychol.} \textbf{\bibinfo{volume}{13}}, \bibinfo{pages}{279}
  (\bibinfo{year}{1977}).

\bibitem[{\citenamefont{Tversky and Kahneman}(1973)}]{tversky1973availability}
\bibinfo{author}{\bibfnamefont{A.}~\bibnamefont{Tversky}} \bibnamefont{and}
  \bibinfo{author}{\bibfnamefont{D.}~\bibnamefont{Kahneman}},
  \bibinfo{journal}{Cognitive psychology} \textbf{\bibinfo{volume}{5}},
  \bibinfo{pages}{207} (\bibinfo{year}{1973}).

\bibitem[{\citenamefont{Gerbner et~al.}(1986)\citenamefont{Gerbner, Gross,
  Morgan, and Signorielli}}]{gerbner1986living}
\bibinfo{author}{\bibfnamefont{G.}~\bibnamefont{Gerbner}},
  \bibinfo{author}{\bibfnamefont{L.}~\bibnamefont{Gross}},
  \bibinfo{author}{\bibfnamefont{M.}~\bibnamefont{Morgan}}, \bibnamefont{and}
  \bibinfo{author}{\bibfnamefont{N.}~\bibnamefont{Signorielli}},
  \bibinfo{journal}{Perspectives on media effects}
  \textbf{\bibinfo{volume}{1986}}, \bibinfo{pages}{17} (\bibinfo{year}{1986}).

\bibitem[{\citenamefont{Sherman and Van~Boven}(2024)}]{sherman2024connections}
\bibinfo{author}{\bibfnamefont{D.~K.} \bibnamefont{Sherman}} \bibnamefont{and}
  \bibinfo{author}{\bibfnamefont{L.}~\bibnamefont{Van~Boven}},
  \bibinfo{journal}{Social Issues and Policy Review}
  \textbf{\bibinfo{volume}{18}}, \bibinfo{pages}{31} (\bibinfo{year}{2024}).

\bibitem[{\citenamefont{Samuelson and Zeckhauser}(1988)}]{samuelson1988status}
\bibinfo{author}{\bibfnamefont{W.}~\bibnamefont{Samuelson}} \bibnamefont{and}
  \bibinfo{author}{\bibfnamefont{R.}~\bibnamefont{Zeckhauser}},
  \bibinfo{journal}{Journal of risk and uncertainty}
  \textbf{\bibinfo{volume}{1}}, \bibinfo{pages}{7} (\bibinfo{year}{1988}).

\bibitem[{\citenamefont{Boykoff and Boykoff}(2004)}]{boykoff2004balance}
\bibinfo{author}{\bibfnamefont{M.~T.} \bibnamefont{Boykoff}} \bibnamefont{and}
  \bibinfo{author}{\bibfnamefont{J.~M.} \bibnamefont{Boykoff}},
  \bibinfo{journal}{Global environmental change} \textbf{\bibinfo{volume}{14}},
  \bibinfo{pages}{125} (\bibinfo{year}{2004}).

\bibitem[{\citenamefont{McAllister et~al.}(2021)\citenamefont{McAllister, Daly,
  Chandler, McNatt, Benham, and Boykoff}}]{mcallister2021balance}
\bibinfo{author}{\bibfnamefont{L.}~\bibnamefont{McAllister}},
  \bibinfo{author}{\bibfnamefont{M.}~\bibnamefont{Daly}},
  \bibinfo{author}{\bibfnamefont{P.}~\bibnamefont{Chandler}},
  \bibinfo{author}{\bibfnamefont{M.}~\bibnamefont{McNatt}},
  \bibinfo{author}{\bibfnamefont{A.}~\bibnamefont{Benham}}, \bibnamefont{and}
  \bibinfo{author}{\bibfnamefont{M.}~\bibnamefont{Boykoff}},
  \bibinfo{journal}{Environmental Research Letters}
  \textbf{\bibinfo{volume}{16}}, \bibinfo{pages}{094008}
  (\bibinfo{year}{2021}).

\bibitem[{\citenamefont{Landgren et~al.}(2026)\citenamefont{Landgren,
  Osborne-Gowey, Garland, Boykoff, and Burgess}}]{landgren2026mediabalance}
\bibinfo{author}{\bibfnamefont{E.}~\bibnamefont{Landgren}},
  \bibinfo{author}{\bibfnamefont{J.}~\bibnamefont{Osborne-Gowey}},
  \bibinfo{author}{\bibfnamefont{J.}~\bibnamefont{Garland}},
  \bibinfo{author}{\bibfnamefont{M.~T.} \bibnamefont{Boykoff}},
  \bibnamefont{and} \bibinfo{author}{\bibfnamefont{M.~G.}
  \bibnamefont{Burgess}}, \bibinfo{journal}{Environmental Research
  Communications}  (\bibinfo{year}{2026}).

\bibitem[{\citenamefont{Chinn et~al.}(2020)\citenamefont{Chinn, Hart, and
  Soroka}}]{chinn2020politicization}
\bibinfo{author}{\bibfnamefont{S.}~\bibnamefont{Chinn}},
  \bibinfo{author}{\bibfnamefont{P.~S.} \bibnamefont{Hart}}, \bibnamefont{and}
  \bibinfo{author}{\bibfnamefont{S.}~\bibnamefont{Soroka}},
  \bibinfo{journal}{Science Communication} \textbf{\bibinfo{volume}{42}},
  \bibinfo{pages}{112} (\bibinfo{year}{2020}).

\bibitem[{\citenamefont{Garrett et~al.}(2019)\citenamefont{Garrett, Long, and
  Jeong}}]{garrett2019partisan}
\bibinfo{author}{\bibfnamefont{R.~K.} \bibnamefont{Garrett}},
  \bibinfo{author}{\bibfnamefont{J.~A.} \bibnamefont{Long}}, \bibnamefont{and}
  \bibinfo{author}{\bibfnamefont{M.~S.} \bibnamefont{Jeong}},
  \bibinfo{journal}{Journal of Communication} \textbf{\bibinfo{volume}{69}},
  \bibinfo{pages}{490} (\bibinfo{year}{2019}).

\bibitem[{\citenamefont{Guay et~al.}(2025)\citenamefont{Guay, Marghetis, Wong,
  and Landy}}]{guay2025quirks}
\bibinfo{author}{\bibfnamefont{B.}~\bibnamefont{Guay}},
  \bibinfo{author}{\bibfnamefont{T.}~\bibnamefont{Marghetis}},
  \bibinfo{author}{\bibfnamefont{C.}~\bibnamefont{Wong}}, \bibnamefont{and}
  \bibinfo{author}{\bibfnamefont{D.}~\bibnamefont{Landy}},
  \bibinfo{journal}{Proceedings of the National Academy of Sciences}
  \textbf{\bibinfo{volume}{122}}, \bibinfo{pages}{e2413064122}
  (\bibinfo{year}{2025}).

\bibitem[{\citenamefont{Mernyk et~al.}(2022)\citenamefont{Mernyk, Pink,
  Druckman, and Willer}}]{mernyk2022correcting}
\bibinfo{author}{\bibfnamefont{J.~S.} \bibnamefont{Mernyk}},
  \bibinfo{author}{\bibfnamefont{S.~L.} \bibnamefont{Pink}},
  \bibinfo{author}{\bibfnamefont{J.~N.} \bibnamefont{Druckman}},
  \bibnamefont{and} \bibinfo{author}{\bibfnamefont{R.}~\bibnamefont{Willer}},
  \bibinfo{journal}{Proceedings of the National Academy of Sciences}
  \textbf{\bibinfo{volume}{119}}, \bibinfo{pages}{e2116851119}
  (\bibinfo{year}{2022}).

\bibitem[{\citenamefont{Dixon et~al.}(2020)\citenamefont{Dixon, Garrett,
  Susmann, and Bushman}}]{dixon2020public}
\bibinfo{author}{\bibfnamefont{G.}~\bibnamefont{Dixon}},
  \bibinfo{author}{\bibfnamefont{K.}~\bibnamefont{Garrett}},
  \bibinfo{author}{\bibfnamefont{M.}~\bibnamefont{Susmann}}, \bibnamefont{and}
  \bibinfo{author}{\bibfnamefont{B.~J.} \bibnamefont{Bushman}},
  \bibinfo{journal}{JAMA network open} \textbf{\bibinfo{volume}{3}},
  \bibinfo{pages}{e2029571} (\bibinfo{year}{2020}).

\bibitem[{\citenamefont{Colleoni et~al.}(2014)\citenamefont{Colleoni, Rozza,
  and Arvidsson}}]{colleoni2014echo}
\bibinfo{author}{\bibfnamefont{E.}~\bibnamefont{Colleoni}},
  \bibinfo{author}{\bibfnamefont{A.}~\bibnamefont{Rozza}}, \bibnamefont{and}
  \bibinfo{author}{\bibfnamefont{A.}~\bibnamefont{Arvidsson}},
  \bibinfo{journal}{Journal of communication} \textbf{\bibinfo{volume}{64}},
  \bibinfo{pages}{317} (\bibinfo{year}{2014}).

\bibitem[{\citenamefont{Boutyline and Willer}(2017)}]{boutyline2017social}
\bibinfo{author}{\bibfnamefont{A.}~\bibnamefont{Boutyline}} \bibnamefont{and}
  \bibinfo{author}{\bibfnamefont{R.}~\bibnamefont{Willer}},
  \bibinfo{journal}{Political psychology} \textbf{\bibinfo{volume}{38}},
  \bibinfo{pages}{551} (\bibinfo{year}{2017}).

\bibitem[{\citenamefont{Brown et~al.}(2025)\citenamefont{Brown, Ventura,
  Tucker, and Nagler}}]{brown2025relationship}
\bibinfo{author}{\bibfnamefont{M.~A.} \bibnamefont{Brown}},
  \bibinfo{author}{\bibfnamefont{T.}~\bibnamefont{Ventura}},
  \bibinfo{author}{\bibfnamefont{J.~A.} \bibnamefont{Tucker}},
  \bibnamefont{and} \bibinfo{author}{\bibfnamefont{J.}~\bibnamefont{Nagler}},
  \bibinfo{journal}{arXiv preprint arXiv:2512.07121}  (\bibinfo{year}{2025}).

\bibitem[{\citenamefont{Feldman et~al.}(2012)\citenamefont{Feldman, Maibach,
  Roser-Renouf, and Leiserowitz}}]{feldman2012climate}
\bibinfo{author}{\bibfnamefont{L.}~\bibnamefont{Feldman}},
  \bibinfo{author}{\bibfnamefont{E.~W.} \bibnamefont{Maibach}},
  \bibinfo{author}{\bibfnamefont{C.}~\bibnamefont{Roser-Renouf}},
  \bibnamefont{and}
  \bibinfo{author}{\bibfnamefont{A.}~\bibnamefont{Leiserowitz}},
  \bibinfo{journal}{The International Journal of Press/Politics}
  \textbf{\bibinfo{volume}{17}}, \bibinfo{pages}{3} (\bibinfo{year}{2012}).

\bibitem[{\citenamefont{Karimi et~al.}(2018)\citenamefont{Karimi, G{\'e}nois,
  Wagner, Singer, and Strohmaier}}]{karimi2018homophily}
\bibinfo{author}{\bibfnamefont{F.}~\bibnamefont{Karimi}},
  \bibinfo{author}{\bibfnamefont{M.}~\bibnamefont{G{\'e}nois}},
  \bibinfo{author}{\bibfnamefont{C.}~\bibnamefont{Wagner}},
  \bibinfo{author}{\bibfnamefont{P.}~\bibnamefont{Singer}}, \bibnamefont{and}
  \bibinfo{author}{\bibfnamefont{M.}~\bibnamefont{Strohmaier}},
  \bibinfo{journal}{Scientific reports} \textbf{\bibinfo{volume}{8}},
  \bibinfo{pages}{11077} (\bibinfo{year}{2018}).

\bibitem[{\citenamefont{Barab{\'a}si and Albert}(1999)}]{barabasi1999emergence}
\bibinfo{author}{\bibfnamefont{A.-L.} \bibnamefont{Barab{\'a}si}}
  \bibnamefont{and} \bibinfo{author}{\bibfnamefont{R.}~\bibnamefont{Albert}},
  \bibinfo{journal}{Science} \textbf{\bibinfo{volume}{286}},
  \bibinfo{pages}{509} (\bibinfo{year}{1999}).

\bibitem[{\citenamefont{Robertson et~al.}(2024)\citenamefont{Robertson,
  Del~Rosario, and Van~Bavel}}]{robertson2024inside}
\bibinfo{author}{\bibfnamefont{C.~E.} \bibnamefont{Robertson}},
  \bibinfo{author}{\bibfnamefont{K.~S.} \bibnamefont{Del~Rosario}},
  \bibnamefont{and} \bibinfo{author}{\bibfnamefont{J.~J.}
  \bibnamefont{Van~Bavel}}, \bibinfo{journal}{Current Opinion in Psychology}
  \textbf{\bibinfo{volume}{60}}, \bibinfo{pages}{101918}
  (\bibinfo{year}{2024}).

\bibitem[{\citenamefont{Adamic and Glance}(2005)}]{adamic2005political}
\bibinfo{author}{\bibfnamefont{L.~A.} \bibnamefont{Adamic}} \bibnamefont{and}
  \bibinfo{author}{\bibfnamefont{N.}~\bibnamefont{Glance}}, in
  \emph{\bibinfo{booktitle}{Proceedings of the 3rd international workshop on
  Link discovery}} (\bibinfo{year}{2005}), pp. \bibinfo{pages}{36--43}.

\bibitem[{\citenamefont{Rogers and Jost}(2022)}]{rogers2022liberals}
\bibinfo{author}{\bibfnamefont{N.}~\bibnamefont{Rogers}} \bibnamefont{and}
  \bibinfo{author}{\bibfnamefont{J.~T.} \bibnamefont{Jost}},
  \bibinfo{journal}{Journal of the Association for Consumer Research}
  \textbf{\bibinfo{volume}{7}}, \bibinfo{pages}{255} (\bibinfo{year}{2022}).

\bibitem[{\citenamefont{Chang et~al.}(2025)\citenamefont{Chang, Druckman,
  Ferrara, and Willer}}]{chang2025liberals}
\bibinfo{author}{\bibfnamefont{H.-C.~H.} \bibnamefont{Chang}},
  \bibinfo{author}{\bibfnamefont{J.~N.} \bibnamefont{Druckman}},
  \bibinfo{author}{\bibfnamefont{E.}~\bibnamefont{Ferrara}}, \bibnamefont{and}
  \bibinfo{author}{\bibfnamefont{R.}~\bibnamefont{Willer}},
  \bibinfo{journal}{PNAS nexus} \textbf{\bibinfo{volume}{4}},
  \bibinfo{pages}{pgaf206} (\bibinfo{year}{2025}).

\bibitem[{\citenamefont{Martel et~al.}(2024)\citenamefont{Martel, Mosleh, Yang,
  Zaman, and Rand}}]{martel2024blocking}
\bibinfo{author}{\bibfnamefont{C.}~\bibnamefont{Martel}},
  \bibinfo{author}{\bibfnamefont{M.}~\bibnamefont{Mosleh}},
  \bibinfo{author}{\bibfnamefont{Q.}~\bibnamefont{Yang}},
  \bibinfo{author}{\bibfnamefont{T.}~\bibnamefont{Zaman}}, \bibnamefont{and}
  \bibinfo{author}{\bibfnamefont{D.~G.} \bibnamefont{Rand}},
  \bibinfo{journal}{PNAS nexus} \textbf{\bibinfo{volume}{3}},
  \bibinfo{pages}{pgae161} (\bibinfo{year}{2024}).

\bibitem[{\citenamefont{Norman and Green}(2025)}]{norman2025can}
\bibinfo{author}{\bibfnamefont{J.~M.} \bibnamefont{Norman}} \bibnamefont{and}
  \bibinfo{author}{\bibfnamefont{B.}~\bibnamefont{Green}},
  \bibinfo{journal}{Political Psychology} \textbf{\bibinfo{volume}{46}},
  \bibinfo{pages}{1364} (\bibinfo{year}{2025}).

\bibitem[{\citenamefont{Sleiman et~al.}(2025)\citenamefont{Sleiman, Melios, and
  Dolan}}]{sleiman2025sleeping}
\bibinfo{author}{\bibfnamefont{Y.}~\bibnamefont{Sleiman}},
  \bibinfo{author}{\bibfnamefont{G.}~\bibnamefont{Melios}}, \bibnamefont{and}
  \bibinfo{author}{\bibfnamefont{P.}~\bibnamefont{Dolan}},
  \bibinfo{journal}{Political Science Research and Methods} pp.
  \bibinfo{pages}{1--23} (\bibinfo{year}{2025}).

\bibitem[{\citenamefont{Zimmaro and Olsson}(2025)}]{zimmaro2025meta}
\bibinfo{author}{\bibfnamefont{F.}~\bibnamefont{Zimmaro}} \bibnamefont{and}
  \bibinfo{author}{\bibfnamefont{H.}~\bibnamefont{Olsson}},
  \bibinfo{journal}{arXiv preprint arXiv:2502.14362}  (\bibinfo{year}{2025}).

\bibitem[{\citenamefont{Dalege et~al.}(2025)\citenamefont{Dalege, Galesic, and
  Olsson}}]{dalege2025networks}
\bibinfo{author}{\bibfnamefont{J.}~\bibnamefont{Dalege}},
  \bibinfo{author}{\bibfnamefont{M.}~\bibnamefont{Galesic}}, \bibnamefont{and}
  \bibinfo{author}{\bibfnamefont{H.}~\bibnamefont{Olsson}},
  \bibinfo{journal}{Psychological review} \textbf{\bibinfo{volume}{132}},
  \bibinfo{pages}{253} (\bibinfo{year}{2025}).

\bibitem[{\citenamefont{Diamond and Zhou}(2022)}]{diamond2022whose}
\bibinfo{author}{\bibfnamefont{E.}~\bibnamefont{Diamond}} \bibnamefont{and}
  \bibinfo{author}{\bibfnamefont{J.}~\bibnamefont{Zhou}},
  \bibinfo{journal}{Environmental Politics} \textbf{\bibinfo{volume}{31}},
  \bibinfo{pages}{991} (\bibinfo{year}{2022}).

\bibitem[{\citenamefont{Van~Boven et~al.}(2018)\citenamefont{Van~Boven, Ehret,
  and Sherman}}]{van2018psychological}
\bibinfo{author}{\bibfnamefont{L.}~\bibnamefont{Van~Boven}},
  \bibinfo{author}{\bibfnamefont{P.~J.} \bibnamefont{Ehret}}, \bibnamefont{and}
  \bibinfo{author}{\bibfnamefont{D.~K.} \bibnamefont{Sherman}},
  \bibinfo{journal}{Perspectives on Psychological Science}
  \textbf{\bibinfo{volume}{13}}, \bibinfo{pages}{492} (\bibinfo{year}{2018}).

\bibitem[{\citenamefont{Vesely and Kl{\"o}ckner}(2020)}]{vesely2020social}
\bibinfo{author}{\bibfnamefont{S.}~\bibnamefont{Vesely}} \bibnamefont{and}
  \bibinfo{author}{\bibfnamefont{C.~A.} \bibnamefont{Kl{\"o}ckner}},
  \bibinfo{journal}{Frontiers in psychology} \textbf{\bibinfo{volume}{11}},
  \bibinfo{pages}{1395} (\bibinfo{year}{2020}).

\bibitem[{\citenamefont{Mastroianni and
  Dana}(2022)}]{mastroianni2022widespread}
\bibinfo{author}{\bibfnamefont{A.~M.} \bibnamefont{Mastroianni}}
  \bibnamefont{and} \bibinfo{author}{\bibfnamefont{J.}~\bibnamefont{Dana}},
  \bibinfo{journal}{Proceedings of the National Academy of Sciences}
  \textbf{\bibinfo{volume}{119}}, \bibinfo{pages}{e2107260119}
  (\bibinfo{year}{2022}).

\bibitem[{\citenamefont{Lee et~al.}(2019)\citenamefont{Lee, Karimi, Wagner, Jo,
  Strohmaier, and Galesic}}]{lee2019homophily}
\bibinfo{author}{\bibfnamefont{E.}~\bibnamefont{Lee}},
  \bibinfo{author}{\bibfnamefont{F.}~\bibnamefont{Karimi}},
  \bibinfo{author}{\bibfnamefont{C.}~\bibnamefont{Wagner}},
  \bibinfo{author}{\bibfnamefont{H.-H.} \bibnamefont{Jo}},
  \bibinfo{author}{\bibfnamefont{M.}~\bibnamefont{Strohmaier}},
  \bibnamefont{and} \bibinfo{author}{\bibfnamefont{M.}~\bibnamefont{Galesic}},
  \bibinfo{journal}{Nat. Hum. Behav.} \textbf{\bibinfo{volume}{3}},
  \bibinfo{pages}{1078} (\bibinfo{year}{2019}).

\end{thebibliography}

\clearpage
\onecolumngrid

\begin{center}
\Large \textbf{Supporting Information for\\
Can homophily explain public underestimation of climate policy support?}
\end{center}

\FloatBarrier

\renewcommand{\thefigure}{S\arabic{figure}}
\setcounter{figure}{0}

\renewcommand{\thetable}{S\arabic{table}}
\setcounter{table}{0}

\renewcommand{\theequation}{S\arabic{equation}}
\setcounter{equation}{0}

\renewcommand{\thepage}{S\arabic{page}} 
\setcounter{page}{1}

\section{Survey Data and Additional Empirical Results}
This present study contains a partial reanalysis of data from a previous survey study by Sparkman et al.
The survey participants are a large sample of U.S.\ adults ($N=6{,}119$) between April and May 2021 recruited through the Ipsos eNation Omnibus nationally representative panel.

A subset of the participants ($N=5{,}375$) provided responses on media consumption habits, used in main text Fig.~\ref{fig:histograms}D.

SI Table~\ref{table:all-demographics} extends the carbon tax data presented Fig.~\ref{fig:histograms} to four additional climate issues and a broader set of demographic groups.
Misperception $\beta_{\text{group}}$ is substantial across all policies and groups, ranging from approximately 17\% (Democrats, worried about climate) to 42\% (Republicans who watch Fox News, siting RE).
Republicans and Fox News viewers exhibit higher misperception than Democrats and independents, and regular Fox News consumption is associated with increased misperception within each partisan group.
Siting renewable energy on public lands shows the largest misperception overall (actual support: 80\%), while worry about climate change shows the smallest.

\begin{table*}[hbtp]
\caption{Misperception $\beta$, based on the survey~\cite{sparkman2022americans}, computed for several demographic groups.}
\centering
\begin{tabular}{llllll}
                                                & \multicolumn{5}{l}{Misperception $\beta_{\text{group}}$ of support for position or policy}                                                                     \\ \cline{2-6}
\multicolumn{1}{l|}{}                           & \multicolumn{1}{l|}{Climate worry} & \multicolumn{1}{l|}{Carbon tax} & \multicolumn{1}{l|}{100\% RE} & \multicolumn{1}{l|}{Siting RE} & \multicolumn{1}{l|}{GND}   \\ \hline
\multicolumn{1}{|l|}{All participants}           & \multicolumn{1}{l|}{21.35\%}   & \multicolumn{1}{l|}{30.15\%}      & \multicolumn{1}{l|}{26.24\%}    & \multicolumn{1}{l|}{35.94\%}     & \multicolumn{1}{l|}{27.94\%}\\ \hline
\multicolumn{1}{|l|}{Democrats}                  & \multicolumn{1}{l|}{17.40\%}    & \multicolumn{1}{l|}{27.03\%}      & \multicolumn{1}{l|}{23.04\%}    & \multicolumn{1}{l|}{34.13\%}     & \multicolumn{1}{l|}{23.26\%} \\
\multicolumn{1}{|l|}{Democrats who watch Fox}         & \multicolumn{1}{l|}{28.03\%}   & \multicolumn{1}{l|}{32.58\%}      & \multicolumn{1}{l|}{26.31\%}    & \multicolumn{1}{l|}{39.95\%}     & \multicolumn{1}{l|}{26.54\%} \\ \hline
\multicolumn{1}{|l|}{Republicans}                & \multicolumn{1}{l|}{28.08\%}   & \multicolumn{1}{l|}{35.12\%}      & \multicolumn{1}{l|}{31.10\%}     & \multicolumn{1}{l|}{39.95\%}     & \multicolumn{1}{l|}{34.66\%} \\
\multicolumn{1}{|l|}{Republicans who watch Fox}  & \multicolumn{1}{l|}{30.20\%}    & \multicolumn{1}{l|}{35.78\%}      & \multicolumn{1}{l|}{32.21\%}    & \multicolumn{1}{l|}{41.79\%}     & \multicolumn{1}{l|}{36.16\%} \\ \hline
\multicolumn{1}{|l|}{Independents}               & \multicolumn{1}{l|}{19.59\%}   & \multicolumn{1}{l|}{29.09\%}      & \multicolumn{1}{l|}{25.93\%}    & \multicolumn{1}{l|}{34.42\%}     & \multicolumn{1}{l|}{28.00\%}    \\
\multicolumn{1}{|l|}{Independents who watch Fox} & \multicolumn{1}{l|}{25.09\%}   & \multicolumn{1}{l|}{32.12\%}      & \multicolumn{1}{l|}{28.53\%}    & \multicolumn{1}{l|}{38.24\%}     & \multicolumn{1}{l|}{32.41\%} \\ \hline
\end{tabular}
\label{table:all-demographics}
\end{table*}

Fig.~\ref{fig:correlation} shows Pearson correlations between binary media consumption (watching a given outlet at least once per week) and group-level misperception $\beta_\text{group}$.
Media outlets are ordered from most liberal-leaning (left) to most conservative-leaning (right).
Positive values (red) indicate that more frequent consumption of a given outlet is associated with greater underestimation of public support; negative values (blue) indicate the opposite.
Conservative-leaning outlets (Fox News, Other Conservative) show the strongest positive associations with misperception across all five issues, while liberal-leaning outlets show weak negative associations.

\begin{figure}[hbtp]
    \centering
    \includegraphics[width=\linewidth]{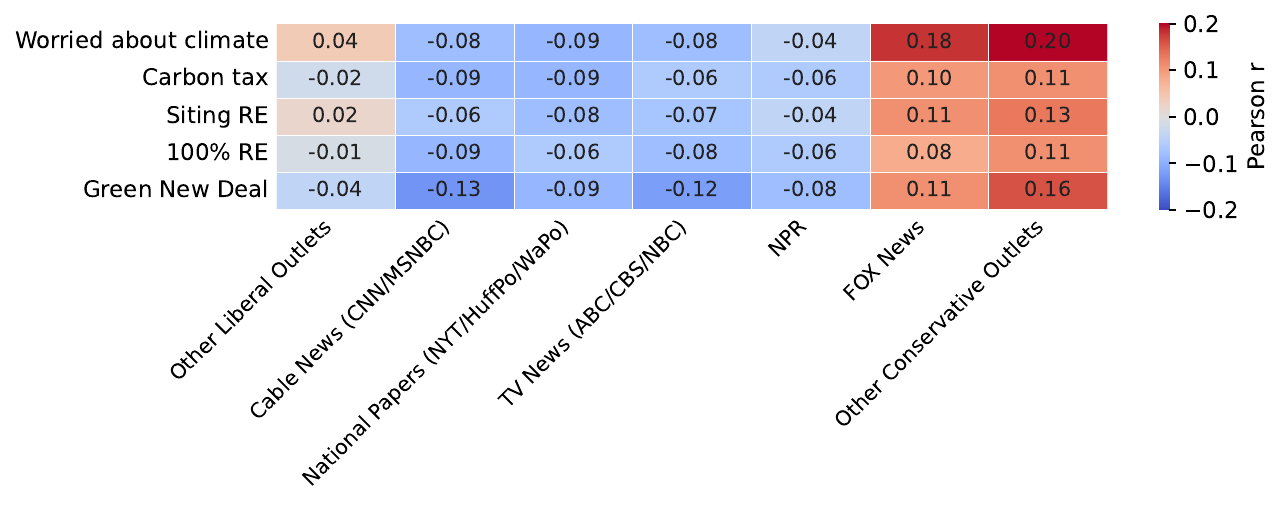}
    \caption{Pearson correlations between binary media consumption (watching a given outlet
    at least once per week) and group-level misperception $\beta_{\text{group}}$ for five
    climate positions, from Sparkman et al.'s survey \cite{sparkman2022americans}.
    Media outlets are ordered from most liberal-leaning (left) to most conservative-leaning (right).
    Positive values (red) indicate that more frequent consumption of a given outlet is associated
    with greater underestimation of public support; negative values (blue) indicate the opposite.
    Conservative-leaning outlets (Fox News, Other Conservative) show the strongest positive
    associations with misperception across all five issues, while liberal-leaning outlets
    show weak negative associations.}
    \label{fig:correlation}
\end{figure}

\section{Convergence of Misperception in Stochastic Block Models}
\label{sec:Extended-Proofs-for-SBM}
In this section, we provide the proofs of the convergence results stated in Theorem 1. We first analyze the misperception measures associated with the deterministic weighted network $\bar{W}^N$ (Lemma \ref{averaged_infinite}), and then establish the almost sure convergence of the corresponding measures on the sampled stochastic block model network $A^N$ (Lemma \ref{sampled_averaged}). For notational convenience, we set $A_{ii}^N=\bar{W}_{ii}^N=1$ for all $i$, thereby incorporating each individual's own opinion through the diagonal entries of the network matrix. This is equivalent to the formulation of misperception used in the main text. 

\begin{lemma}\label{averaged_infinite}
     $\lim_{N\rightarrow{\infty}} \left|\beta_{V}^{}(\bar{W}^{N})-G(f_s,H_{\text{in}}^{s},H_{\text{in}}^{o},H_{\text{out}})\right|=0$ where $$G=G(f_s,H_{\text{in}}^{s},H_{\text{in}}^{o},H_{\text{out}})=f_s - \left((f_s) \frac{f_sH_{\text{in}}^{s}}{f_sH_{\text{in}}^{s}+(1-f_s)H_{\text{out}}} + (1-f_s)\frac{f_sH_{\text{out}}}{f_sH_{\text{out}} + (1-f_s)H_{\text{in}}^{o}}\right).$$
\end{lemma}
\begin{proof} 
Suppose $i\in{S}$. Then, 
\begin{align*}
\frac{1}{\sum_{j=1}^{N} \bar{W}^{N}_{ij}} \sum_{j=1}^{N} \bar{W}^{N}_{ij} \, s(j) 
& = \frac{(\lfloor {f_s}N\rfloor\ -1) H_{\text{in}}^{s} +1 }{(\lfloor {f_s}N\rfloor\ -1) H_{\text{in}}^{s} +1+  \left(N-\lfloor {f_s}N\rfloor\ \right)H_{\text{out}}}.
\end{align*}

\noindent There exists an $\epsilon_N\in(-1,0]$ such that $\lfloor {f_s}N\rfloor = {f_s}N+\epsilon_N$ (see Lemma \ref{rewriting_floor_function} in the Supporting Information). Substituting $\lfloor {f_s}N\rfloor = {f_s}N+\epsilon_N$ and dividing the numerator and denominator by $N$, we obtain,  

\begin{align*}
\frac{1}{\sum_{j=1}^{N} \bar{W}^{N}_{ij}} \sum_{j=1}^{N} \bar{W}^{N}_{ij} \, s(j) =  \frac{\left(f_s+\epsilon_N/N-1/N\right) H_{\text{in}}^{s}+1/N }{\left(f_s+\epsilon_N/N-1/N\right) H_{\text{in}}^{s}+ 1/N+ \left(1-(f_s+\epsilon_N/N)\right) H_{\text{out}}} = G_1(N)\\
\end{align*}

\noindent Now, suppose $i\in{O}$. Then, 
\begin{align*}
\frac{1}{\sum_{j=1}^{N} \bar{W}^{N}_{ij}} \sum_{j=1}^{N} \bar{W}^{N}_{ij} \, s(j) 
&= \frac{\lfloor {f_s}N\rfloor\  H_{\text{out}} }{\lfloor {f_s}N\rfloor\  H_{\text{out}} + \left(N-\lfloor {f_s}N\rfloor\ -1\right)H_{\text{in}}^{o}+1}\\
&= \frac{ (f_s+\epsilon_{N}/N)   H_{\text{out}} }{(f_s+\epsilon_{N}/N)   H_{\text{out}} + (1-(f_s+\epsilon_{N}/N)-1/N) H_{\text{in}}^{o}} = G_2(N),
\end{align*}

\noindent Now, 
\begin{align*}
 \beta_V^{}(\bar{W}^{N}) &= f_s - \frac{1}{N} \left( |S|G_1(N) +|O| G_2(N)\right)\\
 & =  f_s - \frac{1}{N} \left( \lfloor {f_s}N\rfloor G_1(N) + (N-\lfloor {f_s}N\rfloor) G_2(N)\right)\\
 & = f_s -  \left( ( f_s+\epsilon_N/N )G_1(N) + (1-( f_s+\epsilon_N/N)) G_2(N)\right).
\end{align*}

\noindent Observe, 
\begin{align*}
    &\lim_{N\rightarrow\infty}\beta_V^{}(\bar{W}^{N}) = f_s - \left((f_s) \frac{f_sH_{\text{in}}^{s}}{f_sH_{\text{in}}^{s}+(1-f_s)H_{\text{out}}} + ((1-f_s))\frac{f_sH_{\text{out}}}{f_s H_{\text{out}} + (1-f_s)H_{\text{in}}^{o}}\right)\\ 
    &\implies \lim_{N\rightarrow\infty} \left|\beta_V^{}(\bar{W}^{N}) -\underbrace{f_s - \left((f_s) \frac{f_sH_{\text{in}}^{s}}{f_sH_{\text{in}}^{s}+(1-f_s)H_{\text{out}}} + ((1-f_s))\frac{f_sH_{\text{out}}}{f_s H_{\text{out}} + (1-f_s)H_{\text{in}}^{o}}\right)}_{G}\right|=0.
\end{align*}

\end{proof}

\begin{lemma} \label{sampled_averaged}
   Show $\lim_{N\rightarrow{\infty}} \left|\beta_V^{} (A^{N})- \beta_V^{}(\bar{W}^{N})\right|=0$, with probability 1.
\end{lemma}

\begin{proof} 

\noindent For any $N$, 

\begin{align*}
&\left|\beta_V^{} (A^{N})- \beta_V^{}(\bar{W}^{N})\right|  \\
&=\left|\left(f_s - \frac{1}{N} \sum_{i=1}^{N} 
\frac{1}{\sum_{j=1}^{N} A^{N}_{ij}} 
\sum_{j=1}^{N} A^{N}_{ij} \, s(j)\right)-\left(f_s - \frac{1}{N} \sum_{i=1}^{N} 
\frac{1}{\sum_{j=1}^{N} \bar{W}^{N}_{ij}} 
\sum_{j=1}^{N} \bar{W}^{N}_{ij} \, s(j)\right)\right|\\
&= \left|  \frac{1}{N} \sum_{i=1}^{N} 
\frac{1}{\sum_{j=1}^{N} A^{N}_{ij}} 
\sum_{j=1}^{N} A^{N}_{ij} \, s(j) -\frac{1}{N} \sum_{i=1}^{N} 
\frac{1}{\sum_{j=1}^{N} \bar{W}^{N}_{ij}} 
\sum_{j=1}^{N} \bar{W}^{N}_{ij} \, s(j)\right| \\
&= \left|  \frac{1}{N} \sum_{i=1}^{N} \left(
\frac{\sum_{j=1}^{\lfloor f_s N \rfloor} A^{N}_{ij}}{\sum_{j=1}^{N} A^{N}_{ij}} 
 -\frac{\sum_{j=1}^{\lfloor f_s N \rfloor} \bar{W}^{N}_{ij}}{\sum_{j=1}^{N} \bar{W}^{N}_{ij}} 
\right)\right|\\
&\leq  \frac{1}{N} \sum_{i=1}^{N} \left|
\frac{\sum_{j=1}^{\lfloor f_s N \rfloor} A^{N}_{ij}}{\sum_{j=1}^{N} A^{N}_{ij}} 
 -\frac{\sum_{j=1}^{\lfloor f_s N \rfloor} \bar{W}^{N}_{ij}}{\sum_{j=1}^{N} \bar{W}^{N}_{ij}} \right|\\
&= \frac{1}{N} \sum_{i=1}^{N} \left|
\frac{\sum_{j=1}^{\lfloor f_s N \rfloor} A^{N}_{ij}}{\sum_{j=1}^{N} A^{N}_{ij}} -\frac{\sum_{j=1}^{\lfloor f_s N \rfloor}\bar{W}^{N}_{ij}}{\sum_{j=1}^{N} A^{N}_{ij}} + \frac{\sum_{j=1}^{\lfloor f_s N \rfloor} \bar{W}^{N}_{ij}}{\sum_{j=1}^{N} A^{N}_{ij}} 
 -\frac{\sum_{j=1}^{\lfloor f_s N \rfloor} \bar{W}^{N}_{ij}}{\sum_{j=1}^{N} \bar{W}^{N}_{ij}} \right|\\
 &\leq \frac{1}{N} \sum_{i=1}^{N} \left(\underbrace{\left|\frac{\sum_{j=1}^{\lfloor f_s N \rfloor} A^{N}_{ij}}{\sum_{j=1}^{N} A^{N}_{ij}} -\frac{\sum_{j=1}^{\lfloor f_s N \rfloor}\bar{W}^{N}_{ij}}{\sum_{j=1}^{N} A^{N}_{ij}}\right|}_{I_1} + \underbrace{\left|\sum_{j=1}^{\lfloor f_s N \rfloor} \bar{W}^{N}_{ij}\left(\frac{1}{\sum_{j=1}^{N} A^{N}_{ij}} 
 -\frac{1}{\sum_{j=1}^{N} \bar{W}^{N}_{ij}}\right) \right|}_{I_2}\right)
\end{align*} 

\noindent By Lemma \ref{concentration_inequality} with probability 1, there exists a $N_1$ such that for all $n > N_1$,

$$
\left|\sum_{j = 1}^{\lfloor f_s N\rfloor} A_{ij}^{N} - \sum_{j = 1}^{\lfloor f_s N\rfloor} \bar{W}_{ij}^{N}\right| \leq \lfloor f_s N\rfloor^{3/4}\leq N ^{3/4} \quad \forall i.
$$

\noindent Further, by Lemma \ref{concentration_inequality} with probability 1, there exists an $n_2$ such that for all $n > n_2$,
$$
\left|\sum_{j = 1}^{N} A_{ij}^{N} - \sum_{j = 1}^{N} \bar{W}_{ij}^{N}\right| \leq n ^{3/4} \quad \forall i.
$$

\noindent Choose, $\bar{N} = \max{(N_1,N_2)}$. Suppose $N>{\bar{N}}$, with probability 1, for all $i$,  $I_1\leq \left|  \frac{N^{3/4}}{\sum_{j = 1}^{N} \bar{W}_{ij}^{N}-N^{3/4}} \right|$, and

\begin{align*}
    I_2  &= \left|\sum_{j=1}^{\lfloor {f_s}N\rfloor} \bar{W}^{N}_{ij}\right|\left|\left(\frac{1}{\sum_{j=1}^{N} A^{N}_{ij}} 
 -\frac{1}{\sum_{j=1}^{N} \bar{W}^{N}_{ij}}\right)\right|\\
 &\leq \left|\sum_{j=1}^{n } \bar{W}^{N}_{ij}\right|\left| \frac{\sum_{j=1}^{N} A^{N}_{ij}-\sum_{j=1}^{N} \bar{W}^{N}_{ij}}{\sum_{j=1}^{N} A^{N}_{ij}\sum_{j=1}^{N} \bar{W}^{N}_{ij}}\right|\\
  &\leq \left|\sum_{j=1}^{n } \bar{W}^{N}_{ij}\right|\left| \frac{N^{3/4}}{\left(\sum_{j=1}^{N} \bar{W}^{N}_{ij}-N^{3/4}\right)\sum_{j=1}^{N} \bar{W}^{N}_{ij}}\right|\\
  &\leq \left| \frac{N^{3/4}}{\sum_{j=1}^{N} \bar{W}^{N}_{ij}-N^{3/4}}\right|. 
\end{align*}

\noindent 

\noindent By definition of $\bar{W}_{ij}^{N}$, $\sum_{j = 1}^{N} \bar{W}_{ij}^{N}\geq{c}n$ for some $c>0$. Thus, 

\begin{align*}
    \left|\beta_V^{} (A^{N})- \beta_V^{}(\bar{W}^{N})\right| \leq \frac{2}{n}\sum_{i=1}^{N} \left|\frac{N^{3/4}}{cN-N^{3/4}}\right| = \left|\frac{N^{3/4}}{cN-N^{3/4}}\right|.
\end{align*}

\noindent Since $\lim_{N\rightarrow\infty}\left|\frac{N^{3/4}}{c_{}N-N^{3/4}}\right|={0}$, $\lim_{N\rightarrow{\infty}} \left|\beta_V^{} (A^{N})- \beta_V^{}(\bar{W}^{N})\right|=0$.

\end{proof}

\begin{proof}[Proof of Theorem 1] Part $(a)$ follows from Lemma \ref{averaged_infinite}, Lemma \ref{sampled_averaged}, and the triangle inequality. The proofs for part $(b)$ and $(c)$ are very similar to part $(a)$ and are omitted for brevity.
\end{proof}

\begin{lemma}\label{rewriting_floor_function}
There exists an  $\epsilon\in(-1,0]$ such that $\lfloor x \rfloor = x+\epsilon$.  
\end{lemma}

\begin{proof}
    Observe that $\lfloor  x\rfloor \leq x < \lfloor  x\rfloor+1 \implies 0  \leq x -\lfloor  x\rfloor < 1 \implies 0  \geq \lfloor  x\rfloor-x > -1\implies  0 +x   \geq \lfloor  x\rfloor> -1+x$. Thus, there exists an $\epsilon\in(-1,0]$ such that $\lfloor x \rfloor = x+\epsilon$. 
\end{proof}

\begin{lemma} \label{concentration_inequality}
With probability 1, there exists a $\bar{N}$ such that for all $N > \bar{N}$,

$$
\left|\sum_{j = 1}^{N} A_{ij}^{N} - \sum_{j = 1}^{N} \bar{W}_{ij}^{N}\right| \leq N^{3/4} \quad \forall i.
$$

\end{lemma} 

\begin{proof}

By Hoeffding's inequality,

$$
\mathbb{P}\left[\left|\sum_{j=1}^{N} A_{ij}^{N} - \sum_{j=1}^{N} \bar{W}_{ij}^{N}\right| > N^{3/4} \right] = \mathbb{P}\left(\left|\sum_{j \neq i}\left(A_{i j}^N-\bar{W}_{i j}^N\right)\right|>N^{3 / 4}\right)\leq 2 \exp\left(\frac{-2N^{3/2}}{N-1}\right).
$$

\noindent By the union bound,

$$
\mathbb{P}\left[\left|\sum_{j =1}^{n} A_{ij}^{N} - \sum_{j=1}^{N} \bar{W}_{ij}^{N}\right| \leq N^{3/4} \quad \forall i \right] \geq 1 - 2N \exp\left(\frac{-2N^{3/2}}{N-1}\right).
$$

\noindent Since

$$
\sum_{n=1}^{\infty} 2N \exp\left(\frac{-2N^{3/2}}{N-1}\right) < \infty,
$$

\noindent the Borel–Cantelli Lemma implies that with probability 1, there exists a $\bar{N}$ such that for all $N > \bar{N}$,

$$
\left|\sum_{j = 1}^{N} A_{ij}^{N} - \sum_{j = 1}^{N} \bar{W}_{ij}^{N}\right| \leq N^{3/4} \quad \forall i.
$$

\end{proof}

\section{Derivation of Misperception in the Preferential Attachment Model}
\label{appendix:BA}

The analytical expressions for the group-level misperception have been previously derived in literature~\cite{lee2019homophily}. Here we adopt the expressions from Lee et al. (2019) to show misperception $\beta$, which is the difference between the true fraction of the majority and the perceived fraction, compared to the bias, which is defined as the ratio of one's perception of the minority to the true fraction of the minority. This algebraic transformation allows for ease of contrast between stochastic block models, preferential attachment models, and the empirical data~\cite{sparkman2022americans}.

First, recall that the minority (opponent) homophily parameter $h_{oo}$ and the majority (supporter) homophily parameter $h_{ss}$ governs the structure of  majority\textendash minority and minority\textendash majority  ties as follows:
\begin{align}
h_{so} &=1-h_{ss}, \\
h_{os} &=1-h_{oo}.
\end{align}

Since degree growth in this model follows preferential attachment, the minority growth factor $C$ is a polynomial function that can be derived from the following relationship~\cite{lee2019homophily}:

\begin{equation}
C \;=\; f_o\!\left(1 + \frac{h_{oo}\,C}{h_{oo}\,C + h_{os}\,(2 - C)}\right)
       + f_s\!\left(\frac{h_{so}\,C}{h_{ss}\,(2 - C) + h_{so}\,C}\right),
\end{equation}
where $f_s$ is the fraction of supporters in the network and $f_o = 1 - f_s$ is the fraction of opponents.

We then can express the probability of each of the four types of connection ($p_{oo}$, from opponents to opponents; $p_{ss}$, from supporters to supporters; $p_{os}$, from opponents to supporters; and $p_{so}$, from supporters to opponents) as a function of the homophily parameters above and the minority growth factor $C$:

\begin{align}
p_{oo} &= \frac{h_{oo} \, C}{h_{oo} \, C + h_{os}\,(2 - C)}, \\[1ex]
p_{ss} &= \frac{h_{ss}\,(2 - C)}{h_{ss}\,(2 - C) + h_{so}\,C}, \\[1ex]
p_{os} &= \frac{h_{os}\,(2 - C)}{h_{oo}\,C + h_{os}\,(2 - C)}, \\[1ex]
p_{so} &= \frac{h_{so}\,C}{h_{ss}\,(2 - C) + h_{so}\,C}.
\end{align}

We derive the average fraction of supporters perceived by the opponents $\hat{f}_s^{o}$ and the average fraction of supporters perceived by the supporters $\hat{f}_s^{s}$ as follows:
\begin{align}
\hat{f}_s^{o} &= 
\frac{\tfrac{f_s}{f_o}\,p_{so} + p_{os}}{2p_{oo} + \tfrac{f_s}{f_o}\,p_{so} + p_{os}}, \\[1ex]
\hat{f}_s^{s} &= 
\frac{2\,p_{ss}}{2p_{ss} + \tfrac{f_o}{f_s}\,p_{os} + p_{so}}.
\end{align}

\section{Simulations of Overall Misperception in Preferential Attachment Networks}
\label{sec:si-simulations}

In this section, we explore the behavior of preferential attachment network realizations in Monte Carlo simulations. Since these networks exhibit degree heterogeneity, we expect that the mean field approximation described in the main text
does not capture network behavior fully. Here we show that while the magnitude of overall misperception differs from the mean field approximation, the qualitative trends remain the same. We simulate the preferential attachment networks with Bayesian rescaling for different priors ($\delta=0.5$, $1$, $2$) network sizes ($N=100$, $500$), and densities ($m=2$, $5$, $10$). As the prior $\delta$ increases, the size of the region matching the empirical patterns of misperception shrinks.

\begin{figure*}[t]
\centering
\begin{center}
    \includegraphics[width=0.75\linewidth]{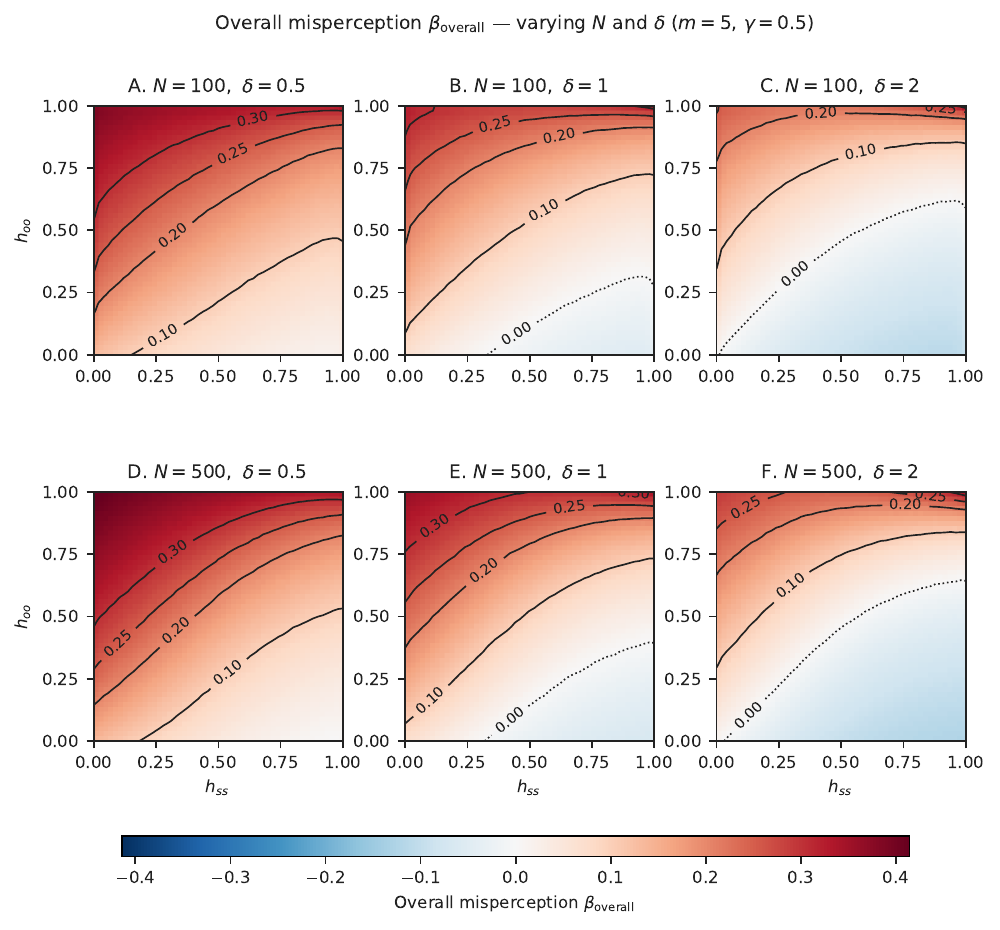}
\end{center}
\caption{Overall misperception $\beta_{\mathrm{overall}}$ as a function of homophily, for varying network sizes and Bayesian priors ($m=5$, $\gamma=0.5$). Each panel shows a heatmap of overall misperception over the grid of supporter homophily $h_{ss}$ and opponent homophily $h_{oo}$ for preferential attachment networks with Bayesian rescaling. The prior varies over columns ($\delta \in \{0.5, 1, 2\}$). The network size varies over the rows ($N \in \{100, 500\}$). Results are averaged across 500 Monte Carlo runs or 100 runs per grid cell.}
\label{fig:si-network-size}
\end{figure*}

\begin{figure*}[t]
\centering
\begin{center}
    \includegraphics[width=0.75\linewidth]{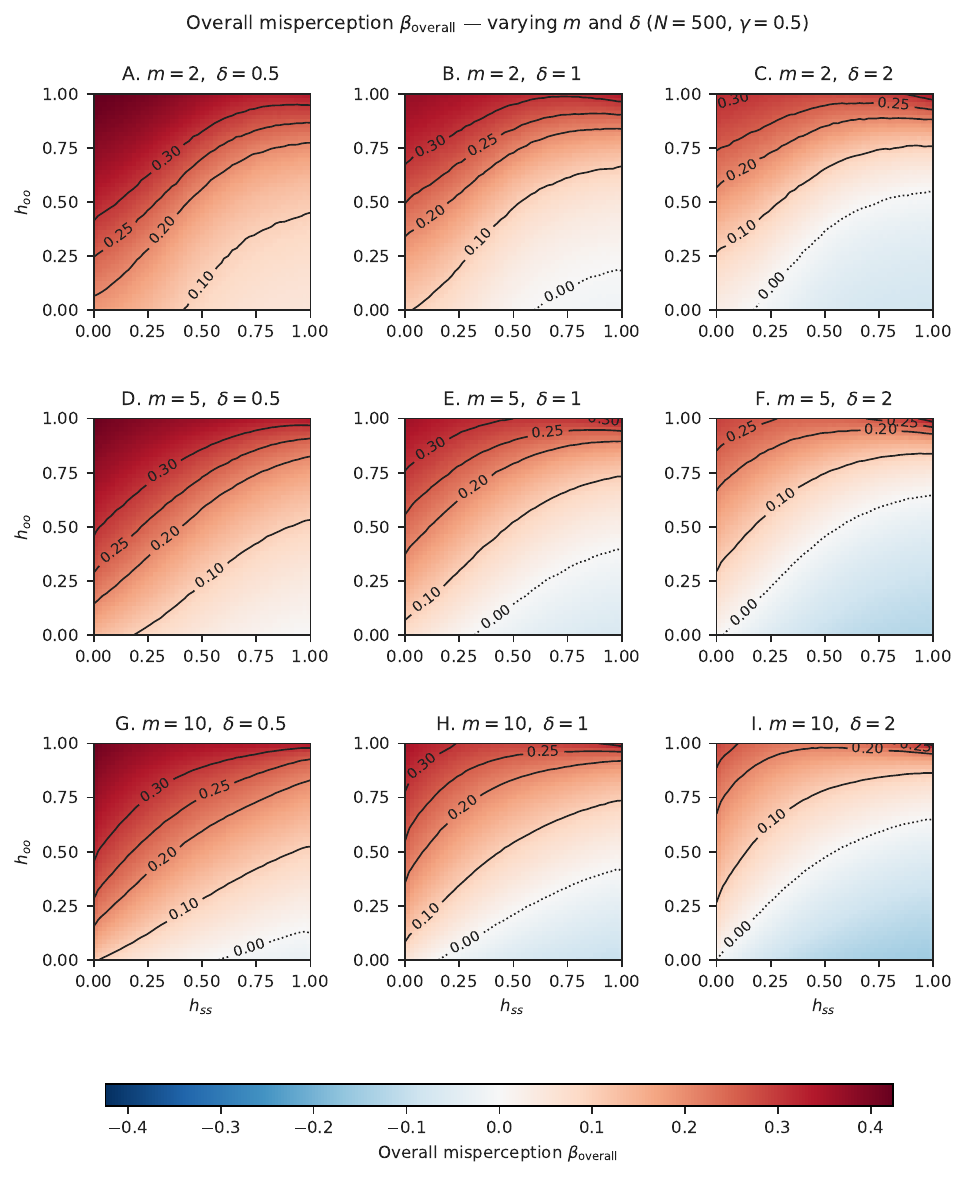}

\end{center}
\caption{Overall misperception $\beta_{\mathrm{overall}}$ as a function of homophily, for varying network densities and Bayesian priors ($N=500$, $\gamma=0.5$ fixed). Each panel shows a heatmap of overall misperception over the grid of supporter homophily $h_{ss}$ and opponent homophily $h_{oo}$, for preferential attachment networks with Bayesian rescaling. The prior varies over columns ($\delta \in \{0.5, 1, 2\}$). The density parameter varies over rows $m \in \{2, 5, 10\}$. The network size is fixed at $N=500$ and the Bayesian uncertainty parameter at $\gamma=0.5$. Contour lines are drawn at $\beta_{\mathrm{overall}} \in \{0.10, 0.20, 0.25, 0.30\}$. Results are averaged across 100 Monte Carlo runs per grid cell.}
\label{fig:si-density}
\end{figure*}

Figure~\ref{fig:si-network-size} shows overall misperception for different network sizes. Similarly to the mean-field approximation, misperception tends to match the empirical pattern when the opponent homophily is higher than the supporter homophily ($h_{oo} > h_{ss}$).
Figure~\ref{fig:si-density} shows overall misperception for different network densities. Again, the misperception is high when opponent homophily exceeds supporter homophily. As networks become denser, the region of parameter space that produces high misperception shrinks. Across these simulations, as the Bayesian prior parameter $\delta$ increases, the region of high misperception shrinks, again matching the mean-field pattern.
Overall, the qualitative trends are consistent with the mean field approximation.

\section{Asymmetric homophily and simulated media bias}

The main text considers the case in which supporters and opponents share the same within-group homophily parameter $h$. Here we relax that assumption, allowing the two groups to have independent homophily parameters $h_{ss}$ (supporter group) and $h_{oo}$ (opponent group). A supporter node attaches to a new neighbor from the same group with probability $h_{ss}$, and an opponent node does so with probability $h_{oo}$. We examine six conditions spanning three qualitative regimes: symmetric homophily ($h_{ss} = h_{oo}$), higher opponent homophily ($h_{oo} > h_{ss}$), and higher supporter homophily ($h_{ss} > h_{oo}$).

Fig.~\ref{fig:asymmetric_swaps} shows misperception $\beta$ as a function of the percentage
of nodes swapped for each of the six conditions, for networks of $N = 1{,}000$ agents with
$m = 5$ and a supporter fraction $f_s = \dfrac{2}{3}$, averaged over 1{,}000 simulated networks.
The gray markers show overall misperception $\beta_{\mathrm{overall}}$; blue and yellow
markers show group-specific misperception for supporters ($\beta_s$) and
opponents ($\beta_o$), respectively.
The shaded vertical band marks the approximate empirical misperception range of 25--30 percentage points
from the survey~\cite{sparkman2022americans}.

Across all six conditions, a moderate number of swaps is sufficient to produce misperception within the empirical range, regardless of whether homophily is symmetric or asymmetric. The minimum swap percentage required is broadly similar across conditions, suggesting that the result is robust to heterogeneity in within-group homophily between the two groups.

\begin{figure*}[t]
\centering
\includegraphics[width=0.5\linewidth]{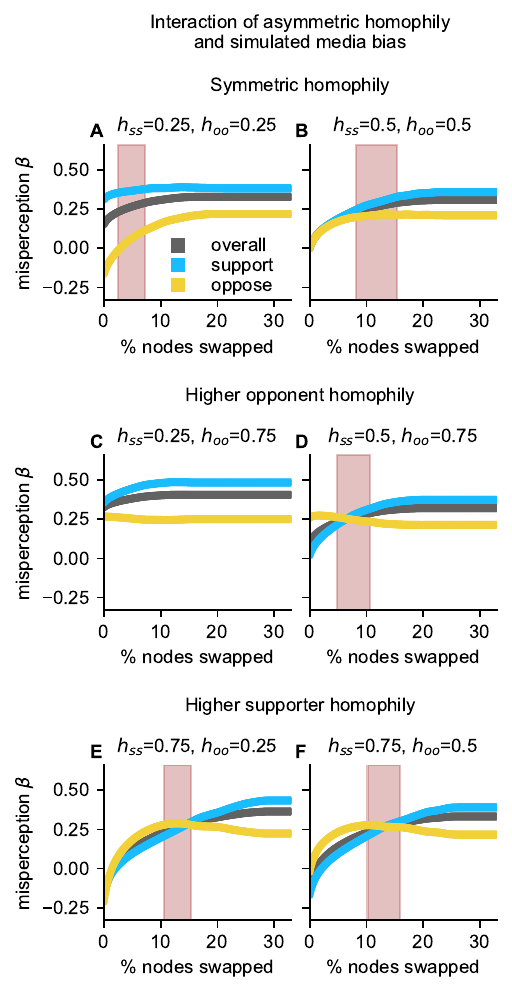}
\caption{Misperception $\beta$ as a function of the percentage of nodes swapped under
    asymmetric homophily, for networks of $N=1{,}000$ agents ($m=5$, supporter fraction
    $f_s=\dfrac{2}{3}$), averaged over 1{,}000 simulations.
    Gray: overall misperception $\beta_{\mathrm{overall}}$.
    Blue: supporters' misperception $\beta_s$.
    Yellow: opponents' misperception $\beta_o$.
    The shaded band marks the approximate empirical misperception range (25--30 percentage points).
    \textbf{(A,~B)} Symmetric homophily ($h_{ss} = h_{oo}$).
    \textbf{(C,~D)} Higher opponent homophily ($h_{oo} > h_{ss}$).
    \textbf{(E,~F)} Higher supporter homophily ($h_{ss} > h_{oo}$). In all conditions, swapping a moderate fraction of nodes drives overall misperception into the empirical range.}
\label{fig:asymmetric_swaps}
\end{figure*}

Fig.~\ref{fig:asymmetric_diagram} summarizes the minimum swap percentage needed to reach
the empirical range across the full $(h_{ss}, h_{oo})$ parameter space.
Each cell shows the minimum swap percentage for that condition; a cross ($\times$) indicates that
misperception is already above the empirical range at zero swaps.
The blue (lower-left) and yellow (upper-right) triangles mark the regions where
$h_{ss} > h_{oo}$ and $h_{oo} > h_{ss}$, respectively; the diagonal marks the symmetric case.
The required swap percentage is similar throughout the parameter space, confirming that the mechanism does not depend on the specific allocation of homophily between groups.

\begin{figure}[t!]
\centering
\includegraphics[width=0.45\linewidth]{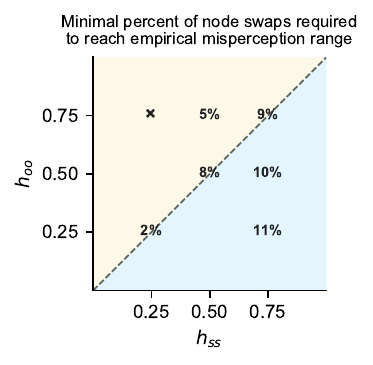}
\caption{Minimum percentage of node swaps required to bring $\beta_{\mathrm{overall}}$
    into the approximate empirical range (25--30 percentage points) for each $(h_{ss}, h_{oo})$ condition
    ($N=1{,}000$, $m=5$, averaged over 1{,}000 simulations).
    A cross ($\times$) indicates that misperception already falls within the empirical range at
    zero swaps.
    Blue triangle (lower-left): $h_{ss} > h_{oo}$ (higher supporter homophily).
    Yellow triangle (upper-right): $h_{oo} > h_{ss}$ (higher opponent homophily).
    The required swap percentage is largely invariant to how homophily is distributed
    between the two groups.}
\label{fig:asymmetric_diagram}
\end{figure}

Tables~\ref{tab:asym_cond_25_25}--\ref{tab:asym_cond_75_50} report misperception
for networks with supporter fraction $f_s = \nicefrac{2}{3}$ across
seven homophily combinations at six swap checkpoints (0, 10, 50, 100, 200, and
300 swaps), for $m = 2$ and $m = 5$.
Each entry is the mean (SD) across 1,000 simulated networks of $N = 1{,}000$ agents.
$\beta_s$ and $\beta_o$ give group-specific misperception for supporters and opponents,
respectively; $\beta_{\mathrm{overall}}$ is the population-weighted average. The standard deviation in misperception does not exceed 3\% in any condition.

\FloatBarrier

\begin{table*}[htbp]
\centering
\caption{Misperception $\beta$ (mean and SD across 1,000 simulations) at each swap checkpoint, for $h_{ss} = 0.25$, $h_{oo} = 0.25$ ($N = 1,000$ agents, supporter fraction $f_s = \nicefrac{2}{3}$). $\beta = f_s \cdot 100 - \hat{p}$ in percentage points; positive values indicate underestimation of support. $\beta_s$: supporter group; $\beta_o$: opponent group; $\beta_{\mathrm{overall}}$: population-weighted average.}
\label{tab:asym_cond_25_25}
\begin{tabular}{rrrrrrr}
\toprule
Swaps & \multicolumn{3}{c}{$m=2$} & \multicolumn{3}{c}{$m=5$} \\
\cmidrule(lr){2-4} \cmidrule(lr){5-7}
 & $\beta_s$ & $\beta_o$ & $\beta_{\mathrm{overall}}$ & $\beta_s$ & $\beta_o$ & $\beta_{\mathrm{overall}}$ \\
\midrule
0 & 32.15 (2.32) & -15.87 (1.85) & 16.11 (1.92) & 31.15 (1.70) & -16.43 (1.17) & 15.26 (1.38) \\
10 & 36.14 (1.88) & -4.97 (1.85) & 22.41 (1.53) & 33.84 (1.44) & -8.34 (1.27) & 19.75 (1.19) \\
50 & 40.76 (1.55) & 12.15 (1.70) & 31.20 (1.23) & 36.69 (1.16) & 6.60 (1.09) & 26.64 (0.93) \\
100 & 44.95 (1.57) & 20.79 (1.26) & 36.88 (1.08) & 38.29 (1.04) & 15.84 (0.86) & 30.79 (0.76) \\
200 & 49.41 (1.18) & 22.27 (0.80) & 40.35 (0.81) & 38.22 (0.95) & 21.85 (0.58) & 32.76 (0.60) \\
300 & 49.33 (1.19) & 22.29 (0.80) & 40.29 (0.82) & 38.19 (0.95) & 21.86 (0.59) & 32.74 (0.60) \\
\bottomrule
\end{tabular}
\end{table*}

\begin{table*}[htbp]
\centering
\caption{Misperception $\beta$ (mean and SD across 1,000 simulations) at each swap checkpoint, for $h_{ss} = 0.5$, $h_{oo} = 0.5$ ($N = 1,000$ agents, supporter fraction $f_s = \nicefrac{2}{3}$). $\beta = f_s \cdot 100 - \hat{p}$ in percentage points; positive values indicate underestimation of support. $\beta_s$: supporter group; $\beta_o$: opponent group; $\beta_{\mathrm{overall}}$: population-weighted average.}
\label{tab:asym_cond_50_50}
\begin{tabular}{rrrrrrr}
\toprule
Swaps & \multicolumn{3}{c}{$m=2$} & \multicolumn{3}{c}{$m=5$} \\
\cmidrule(lr){2-4} \cmidrule(lr){5-7}
 & $\beta_s$ & $\beta_o$ & $\beta_{\mathrm{overall}}$ & $\beta_s$ & $\beta_o$ & $\beta_{\mathrm{overall}}$ \\
\midrule
0 & 0.15 (2.56) & -0.00 (2.94) & 0.10 (2.41) & -0.07 (1.83) & -0.18 (2.00) & -0.11 (1.74) \\
10 & 10.46 (2.08) & 10.16 (2.40) & 10.36 (1.83) & 7.36 (1.58) & 7.15 (1.70) & 7.29 (1.44) \\
50 & 24.24 (1.82) & 21.29 (1.67) & 23.25 (1.38) & 18.82 (1.32) & 17.03 (1.17) & 18.22 (1.03) \\
100 & 34.53 (1.68) & 24.69 (1.19) & 31.24 (1.12) & 27.31 (1.19) & 20.54 (0.82) & 25.04 (0.83) \\
200 & 47.81 (1.48) & 20.88 (0.69) & 38.81 (0.94) & 35.09 (1.01) & 21.27 (0.47) & 30.47 (0.63) \\
300 & 47.95 (1.18) & 20.78 (0.77) & 38.88 (0.79) & 35.78 (1.09) & 20.96 (0.44) & 30.83 (0.67) \\
\bottomrule
\end{tabular}
\end{table*}

\begin{table*}[htbp]
\centering
\caption{Misperception $\beta$ (mean and SD across 1,000 simulations) at each swap checkpoint, for $h_{ss} = 0.75$, $h_{oo} = 0.75$ ($N = 1,000$ agents, supporter fraction $f_s = \nicefrac{2}{3}$). $\beta = f_s \cdot 100 - \hat{p}$ in percentage points; positive values indicate underestimation of support. $\beta_s$: supporter group; $\beta_o$: opponent group; $\beta_{\mathrm{overall}}$: population-weighted average.}
\label{tab:asym_cond_75_75}
\begin{tabular}{rrrrrrr}
\toprule
Swaps & \multicolumn{3}{c}{$m=2$} & \multicolumn{3}{c}{$m=5$} \\
\cmidrule(lr){2-4} \cmidrule(lr){5-7}
 & $\beta_s$ & $\beta_o$ & $\beta_{\mathrm{overall}}$ & $\beta_s$ & $\beta_o$ & $\beta_{\mathrm{overall}}$ \\
\midrule
0 & -18.61 (1.05) & 25.98 (2.76) & -3.72 (1.27) & -18.27 (0.66) & 26.45 (1.93) & -3.33 (0.89) \\
10 & -3.90 (1.64) & 31.84 (2.21) & 8.04 (1.39) & -7.60 (0.93) & 30.28 (1.59) & 5.05 (0.87) \\
50 & 16.03 (1.74) & 34.28 (1.49) & 22.13 (1.26) & 9.74 (1.08) & 31.62 (1.07) & 17.05 (0.77) \\
100 & 30.34 (1.76) & 31.10 (1.08) & 30.59 (1.17) & 22.45 (1.10) & 28.34 (0.74) & 24.42 (0.71) \\
200 & 47.45 (1.76) & 21.64 (0.66) & 38.83 (1.11) & 34.86 (0.97) & 21.83 (0.50) & 30.51 (0.61) \\
300 & 48.17 (1.26) & 20.91 (0.76) & 39.07 (0.86) & 36.05 (1.06) & 20.96 (0.45) & 31.01 (0.64) \\
\bottomrule
\end{tabular}
\end{table*}

\begin{table*}[htbp]
\centering
\caption{Misperception $\beta$ (mean and SD across 1,000 simulations) at each swap checkpoint, for $h_{ss} = 0.25$, $h_{oo} = 0.75$ ($N = 1,000$ agents, supporter fraction $f_s = \nicefrac{2}{3}$). $\beta = f_s \cdot 100 - \hat{p}$ in percentage points; positive values indicate underestimation of support. $\beta_s$: supporter group; $\beta_o$: opponent group; $\beta_{\mathrm{overall}}$: population-weighted average.}
\label{tab:asym_cond_25_75}
\begin{tabular}{rrrrrrr}
\toprule
Swaps & \multicolumn{3}{c}{$m=2$} & \multicolumn{3}{c}{$m=5$} \\
\cmidrule(lr){2-4} \cmidrule(lr){5-7}
 & $\beta_s$ & $\beta_o$ & $\beta_{\mathrm{overall}}$ & $\beta_s$ & $\beta_o$ & $\beta_{\mathrm{overall}}$ \\
\midrule
0 & 37.83 (2.50) & 29.31 (1.81) & 34.98 (1.96) & 35.56 (2.04) & 26.57 (1.17) & 32.56 (1.54) \\
10 & 41.80 (2.01) & 30.02 (1.58) & 37.86 (1.53) & 38.74 (1.69) & 26.55 (1.08) & 34.67 (1.26) \\
50 & 49.07 (1.57) & 29.65 (1.30) & 42.59 (1.16) & 44.62 (1.31) & 25.38 (0.91) & 38.19 (0.95) \\
100 & 55.18 (1.29) & 27.52 (1.13) & 45.94 (0.92) & 48.07 (1.09) & 24.54 (0.78) & 40.21 (0.80) \\
200 & 56.95 (1.38) & 26.77 (0.81) & 46.87 (1.01) & 48.38 (0.95) & 25.01 (0.76) & 40.58 (0.70) \\
300 & 56.95 (1.38) & 26.77 (0.81) & 46.87 (1.01) & 48.38 (0.95) & 25.01 (0.76) & 40.58 (0.70) \\
\bottomrule
\end{tabular}
\end{table*}

\begin{table*}[htbp]
\centering
\caption{Misperception $\beta$ (mean and SD across 1,000 simulations) at each swap checkpoint, for $h_{ss} = 0.75$, $h_{oo} = 0.25$ ($N = 1,000$ agents, supporter fraction $f_s = \nicefrac{2}{3}$). $\beta = f_s \cdot 100 - \hat{p}$ in percentage points; positive values indicate underestimation of support. $\beta_s$: supporter group; $\beta_o$: opponent group; $\beta_{\mathrm{overall}}$: population-weighted average.}
\label{tab:asym_cond_75_25}
\begin{tabular}{rrrrrrr}
\toprule
Swaps & \multicolumn{3}{c}{$m=2$} & \multicolumn{3}{c}{$m=5$} \\
\cmidrule(lr){2-4} \cmidrule(lr){5-7}
 & $\beta_s$ & $\beta_o$ & $\beta_{\mathrm{overall}}$ & $\beta_s$ & $\beta_o$ & $\beta_{\mathrm{overall}}$ \\
\midrule
0 & -17.75 (0.79) & -22.08 (1.70) & -19.19 (0.84) & -16.47 (0.50) & -21.19 (1.17) & -18.05 (0.51) \\
10 & -1.98 (1.70) & 0.65 (2.66) & -1.10 (1.74) & -5.44 (1.01) & -3.88 (1.87) & -4.92 (1.11) \\
50 & 15.43 (1.69) & 22.79 (2.13) & 17.89 (1.44) & 9.56 (1.08) & 18.22 (1.67) & 12.46 (1.02) \\
100 & 27.26 (1.69) & 31.25 (1.53) & 28.59 (1.22) & 19.65 (1.13) & 27.45 (1.18) & 22.25 (0.89) \\
200 & 46.41 (1.66) & 26.94 (0.99) & 39.91 (1.04) & 35.73 (1.09) & 26.56 (0.58) & 32.66 (0.71) \\
300 & 52.84 (1.28) & 23.50 (0.75) & 43.04 (0.89) & 43.43 (1.05) & 22.34 (0.51) & 36.38 (0.68) \\
\bottomrule
\end{tabular}
\end{table*}

\begin{table*}[htbp]
\centering
\caption{Misperception $\beta$ (mean and SD across 1,000 simulations) at each swap checkpoint, for $h_{ss} = 0.5$, $h_{oo} = 0.75$ ($N = 1,000$ agents, supporter fraction $f_s = \nicefrac{2}{3}$). $\beta = f_s \cdot 100 - \hat{p}$ in percentage points; positive values indicate underestimation of support. $\beta_s$: supporter group; $\beta_o$: opponent group; $\beta_{\mathrm{overall}}$: population-weighted average.}
\label{tab:asym_cond_50_75}
\begin{tabular}{rrrrrrr}
\toprule
Swaps & \multicolumn{3}{c}{$m=2$} & \multicolumn{3}{c}{$m=5$} \\
\cmidrule(lr){2-4} \cmidrule(lr){5-7}
 & $\beta_s$ & $\beta_o$ & $\beta_{\mathrm{overall}}$ & $\beta_s$ & $\beta_o$ & $\beta_{\mathrm{overall}}$ \\
\midrule
0 & 3.01 (2.91) & 27.06 (2.44) & 11.05 (2.47) & 1.97 (2.15) & 25.81 (1.58) & 9.93 (1.80) \\
10 & 12.63 (2.29) & 29.62 (1.91) & 18.31 (1.82) & 9.15 (1.72) & 27.17 (1.25) & 15.17 (1.36) \\
50 & 27.59 (1.95) & 29.57 (1.28) & 28.25 (1.36) & 21.89 (1.36) & 26.19 (0.87) & 23.33 (0.94) \\
100 & 38.99 (1.88) & 26.34 (1.06) & 34.77 (1.19) & 30.94 (1.18) & 23.58 (0.64) & 28.48 (0.77) \\
200 & 48.98 (1.23) & 21.56 (0.81) & 39.82 (0.85) & 37.37 (0.94) & 21.35 (0.53) & 32.02 (0.61) \\
300 & 48.93 (1.24) & 21.63 (0.81) & 39.81 (0.85) & 37.37 (0.93) & 21.36 (0.54) & 32.02 (0.60) \\
\bottomrule
\end{tabular}
\end{table*}

\begin{table*}[htbp]
\centering
\caption{Misperception $\beta$ (mean and SD across 1,000 simulations) at each swap checkpoint, for $h_{ss} = 0.75$, $h_{oo} = 0.5$ ($N = 1,000$ agents, supporter fraction $f_s = \nicefrac{2}{3}$). $\beta = f_s \cdot 100 - \hat{p}$ in percentage points; positive values indicate underestimation of support. $\beta_s$: supporter group; $\beta_o$: opponent group; $\beta_{\mathrm{overall}}$: population-weighted average.}
\label{tab:asym_cond_75_50}
\begin{tabular}{rrrrrrr}
\toprule
Swaps & \multicolumn{3}{c}{$m=2$} & \multicolumn{3}{c}{$m=5$} \\
\cmidrule(lr){2-4} \cmidrule(lr){5-7}
 & $\beta_s$ & $\beta_o$ & $\beta_{\mathrm{overall}}$ & $\beta_s$ & $\beta_o$ & $\beta_{\mathrm{overall}}$ \\
\midrule
0 & -17.98 (0.89) & -3.51 (2.75) & -13.14 (1.20) & -16.90 (0.57) & -2.28 (1.97) & -12.02 (0.84) \\
10 & -2.71 (1.75) & 12.17 (2.63) & 2.26 (1.66) & -6.15 (0.97) & 9.28 (1.97) & -1.00 (1.06) \\
50 & 15.52 (1.78) & 26.87 (1.98) & 19.31 (1.42) & 9.48 (1.13) & 23.24 (1.46) & 14.07 (0.94) \\
100 & 28.11 (1.77) & 31.15 (1.39) & 29.12 (1.26) & 20.47 (1.16) & 27.59 (1.00) & 22.85 (0.83) \\
200 & 47.75 (1.66) & 23.58 (0.88) & 39.68 (1.06) & 35.29 (1.11) & 23.96 (0.57) & 31.50 (0.69) \\
300 & 50.05 (1.21) & 22.18 (0.80) & 40.74 (0.85) & 39.03 (0.95) & 21.65 (0.52) & 33.22 (0.62) \\
\bottomrule
\end{tabular}
\end{table*}

\FloatBarrier

\clearpage

\section{Simulation algorithm}
Simulation of swapping the majority and minority nodes.
\begin{figure}[htbp]
\small
\caption{Simulation of swapping the majority and minority nodes}\label{alg:agent_simulation}
\begin{algorithmic}[1]
\State Input: network $G$ of size $N$, minority\_nodes\_list
\State Output: updated\_minority\_nodes\_list

\Comment \textit{Sort nodes of $G$ by degree in descending order}

\State degree\_sorted\_list= list of nodes of $G$ by degree in descending order

\Comment \textit{Initialize flags and indices}

\State flag1 = 0

\State flag2 = 0

\State $i = 0$

\State $j = -1$

\Comment \textit{Change the highest-degree majority node to be a minority node}

\While{flag1 is 0}
\If{node $i$ in degree\_sorted\_list is not in minority\_nodes}
\State flag1 = 1

\State Add $i$ to minority\_nodes\_list

\Else
\State Increment $i$ by 1
 
\EndIf

\EndWhile

\Comment \textit{Change the lowest-degree minority node to be a majority node}

\While{flag2 is 0}
\If{node $j$ in degree\_sorted\_list is in minority\_nodes}

\State flag2 = 1

\State Remove $j$ from minority\_nodes\_list

\Else
\State Decrement $j$ by 1

\EndIf

\EndWhile

\State updated\_minority\_nodes\_list = minority\_nodes\_list \\

\State Return: updated\_minority\_nodes\_list

\end{algorithmic}
\end{figure}

\end{document}